\documentclass[acmsmall]{acmart}
\pdfoutput=1

\usepackage{amsthm}

\usepackage[T1]{fontenc}
\usepackage{color}
\usepackage[english]{babel}
\usepackage{url}

\usepackage{mathtools}
\usepackage{amsmath}

\usepackage{algorithm}
\usepackage{algorithmicx}
\usepackage{array}

\newtheorem{theorem}{\textbf{Theorem}}
\newtheorem*{proof}{\textbf{Proof}}
\usepackage[binary-units=true]{siunitx}
\usepackage{algpseudocode} 
\DeclareMathOperator*{\argmax}{arg\,max}
\DeclareMathOperator*{\argmin}{arg\,min}
\algnewcommand{\algorithmicgoto}{\textbf{go to}}%
\algnewcommand{\Goto}[1]{\algorithmicgoto~\ref{#1}}%
\usepackage{lipsum}
\usepackage{multicol}
\usepackage{graphicx}
\usepackage[title]{appendix}
\theoremstyle{definition}
\newtheorem{definition}{Definition}
\usepackage{color}
\usepackage{soul}

\acmJournal{TOSN}
\acmVolume{9}
\acmNumber{4}
\acmArticle{39}
\acmYear{2010}
\acmMonth{3}
\copyrightyear{2009}

\setcopyright{acmlicensed}

\acmDOI{0000001.0000001}

\received{February 2007}
\received[revised]{March 2009}
\received[accepted]{June 2009}
\newcommand{\ds}[1]{\textcolor{black}{#1}}

\newcommand{\yy}[1]{\textcolor{black}{#1}}

\begin{document}

\title[Game Theory Interference Mitigation  for Time-Division MAC in WBAN]{\ds{Robust Wireless Body Area Networks Coexistence: A Game Theoretic Approach to Time-Division MAC}}

\author{Yizhou Yang}
\orcid{1234-5678-9012-3456}
\affiliation{%
  \institution{Australian National University}
  \city{Canberra, ACT}
  \country{Australia}
  \postcode{2601}}
\email{yizhou.yang@anu.edu.au}
\author{David B.Smith}
\affiliation{%
  \institution{CSIRO, Data61}
  \city{Sydney, NSW}
  \country{Australia}
}
\email{david.smith@data61.csiro.au}

\begin{abstract}
\ds{The enabling of} wireless body area networks (WBANs) \ds{coexistence by radio interference mitigation} is \ds{very} important due to \ds{a} rapid growth \ds{in} potential users, and \ds{a} lack of \ds{a} central coordinator among WBANs \ds{that are closely located}. In this paper, we propose a TDMA based MAC layer Scheme, \ds{with a} back-off mechanism that reduce\ds{s packet} collision probability; \ds{and estimate performance using a} Markov chain model. Based on the \ds{MAC layer} scheme, a novel non-cooperative game \ds{is proposed} to \ds{jointly} adjust sensor node's transmit power and rate. \ds{In comparison with the state-of-art, simulation that includes empirical data} shows that the proposed approach leads to higher throughput and longer node lifespan \ds{as} WBAN wearers dynamically mov\ds{e into each other's} vicinity. Moreover, by adaptively tuning contention windows size \ds{an alternative game} is developed, which significantly reduce\ds{s} the latency. Both \ds{proposed} games provide robust transmission under strong inter-WBAN interferences, but \ds{are demonstrated to be applicable to} different scenarios. The uniqueness and existence of Nash Equilibrium (NE)\ds{, as well as close-to-optimum social efficiency, is} also proven for both games.
\end{abstract}

%
%


%
%


\maketitle

\renewcommand{\shortauthors}{Y. Yang, D. Smith}

\section{Introduction}
%
%
%
%


W\ds{ireless body area networks (WBANs) are an integral part of} affordable, flexible and proactive wearable health-care to reduce cost\ds{s} and improve people's quality of life. Recent advances in wireless communications and sensor hardware \ds{mean that} WBAN\ds{s} are feasible implementations. The IEEE 802.15.6 standard approved in 2012 for wireless communication in WBANs aims to serve a variety of medical, entertain\ds{ment}, military and consumer electronics applications. \ds{IEEE 802.15.6 only outlined} basic requirements \ds{with a choice of} multiple MAC \ds{layer} techniques \ds{--- i.e., scheduled access, polling, and contention access ---} supported in a beacon-based superframe. However, with \ds{a} rapid increase in active devices, which predicted to be \ds{well over a billion} by 2018 \cite{business2013}, WBANs \ds{will} suffer unavoidable inter-WBAN interference from \ds{closely-located} coexisting WBANs \ds{due to no central coordinator amongst networks.}

Moreover, \ds{WBAN radios have limited battery capacity due to sensor/actuator devices small sizes}, but health monitoring applications require long battery life-time\ds{, as removing, charging and replacing batteries can be very inconvenient and difficult}.  Because inter-WBAN interference can cause performance degradation and energy wastage of low-power sensor nodes\ds{, and sensors radio transmit power is strictly constrained, inter-WBAN interference is a major issue, where most energy wastage occurs in the wireless transceiver}. Hence, a well-designed MAC layer protocol \ds{for transmission scheduling} is of paramount importance to prolong network lifetime and improve the robustness of WBAN communications by reducing \ds{periods of interference. I}nterference mitigation schemes have been widely stud\ds{ied} for other networks, such as traditional cellular networks and wireless sensor networks, \ds{but such schemes} can not be directly implemented \ds{in} WBANs. \ds{This is} because WBAN\ds{s} have relatively high mobility\ds{, compared with other networks where gateway devices are typically stationary,  leading} to unique features \ds{of practical} WBAN \ds{coexistence that require new approaches specifically designed for WBAN}s.

Existing literatures on MAC protocols for WBANs demonstrate that CSMA/CA protocols encounter unreliable CCA issues and heavy collision \yy{\cite{5304962}}. On the other hand, TDMA has proven to be more reliable and power efficient\cite{zhang2010performance}. \emph{Therefore, here, we propose a \ds{game-theoretic formulation of a} TDMA-based MAC protocol \ds{to} achieve energy efficiency, reduce inter-BAN interference, improve overall throughput and \ds{reduce} latency \ds{across all} co-existing WBANs.} The proposed method \ds{adapts} to the time-varying channel and traffic by optimizing the transmission schedule. Depending on the special random back-off mechanism to minimize the probability of packets collision among sensors in different \yy{W}BANs, the overall \yy{interference} level can be consequently reduced. Besides, \yy{each WBANs is treated as an active player in a non-cooperative game}.
\yy{In each superframe}, the transmission parameters, such as, transmit power, transmit probability and data rate, \yy{are determined based on a utility function, which admits a unique Nash Equilibrium}. By maximizing the utility function, a higher throughput can be achieved, \ds{and} the latency and power consumption \ds{is} reduced.

\ds{Hence, t}he main contributions of this paper are:
\begin{itemize}
\item[--] A novel MAC layer timing protocol to reduce inter-BAN interference by adapting back-off in TDMA.
\item[--] A Markov chain is constructed to provide performance evaluation.
\item[--] A non-cooperative game is proposed to jointly tune transmit power and data rate to improve throughput and reduce latency.
\item[--] Based on the Markov chain, \yy{an Adaptive Backoff Game} is proposed for better Quality of Service (QoS) performance.
\item[--] The two games \ds{are demonstrated to} reduce \ds{radio} interference by improving throughput\ds{, in conjunction with reduced power and delay.}
\end{itemize}

The proposed method \ds{has two principal features}: (i) a novel MAC layer protocol and (ii) game theor\ds{etic} power control. The MAC layer protocol \ds{focuses} on rescheduling unsuccessfully transmitted packets to reduce the probability of packet collision among different co-located WBANs. \ds{T}wo games \ds{are proposed for power control}: \ds{a} rate-and-power game and \ds{a} Adaptive Backoff Game \ds{(as an extension of the Link Adaptation Game)}. The Link Adaptation Game tuning the node's transmit power and data rate from the Nash equilibrium of its utility function to obtain optimized throughput (in terms of Packer Delivery Ratio, $\mathrm{PDR}$) and power consumption. In the Adaptive Backoff Game, the sensor node adjusts its transmit probability by dynamically chang\ds{ing} the contention parameters (contention window size), to further improve throughput performance and minimize transmission delay. \ds{Due to} the difference in delay performance and power consumption, \ds{a tradeoff of} the two games can be \ds{made according to WBAN application.}

The rest of this paper is organized as follows. The related literature will be reviewed in Section \ref{Sec:2}. The proposed MAC layer scheme, as well as the analytical model, will be described in detail in Section \ref{Sec:3}. The two proposed game theory methods will be depicted in Section \ref{Sec:4} and Section \ref{Sec:5} respectively. The performance of the proposed methods will be illustrated in Section \ref{Sec:6}.

\section{Related Work}\label{Sec:2}
In literature, many studies have been proposed to mitigate inter-WBAN \ds{radio} interference, \ds{in three main categories}:
\begin{itemize}
\item[--] Transmit power control
\item[--] MAC layer scheduling
\item[--] Data rate and power control
\end{itemize}

\subsection{Transmit Power Control}
Several pioneer\ds{ing works o}n inter-WBAN interference mitigation focus on \ds{solutions at the physical layer}. Many techniques involve transmission power control that is based on a centralized and partially distributed \cite{zander1992distributed} approach. These techniques are proven to be effective for networks with \ds{stable topologies and} fewer power constraints\cite{lee1996fully}. However, more recently, \ds{g}ame-theoretic power control, incorporating pricing factors in utility functions, e.g., \cite{zou2014bayesian}, \cite{kazemi2010inter}, has been shown to improve QoS in wireless networks. Due to the general lack of a central coordinator in WBANs, transmit power control must be adopted in a distributed manner across co-located WBANs. In recent studies, WBANs have been modeled as rational players competing for resources in non-cooperative power-control games, \cite{dong2015socially} \cite{7997347}.
\subsection{MAC layer}
Recently, several MAC layer protocols that seek to solve the inter-WBAN interference problems \cite{cheng2013coloring} \cite{movassaghi2014cooperative} \cite{movassaghi2014smart} have been proposed. The work \ds{in} \cite{cheng2013coloring} us\ds{es} cooperative schemes to suppress the inter-WBAN interference, where a random incomplete coloring (RIC) algorithm \ds{is} proposed to realize a fast and high spatial-reuse for inter-WBAN scheduling. In \cite{movassaghi2014cooperative}, a mixed graph coloring is used for interference mitigation among WBANs, \ds{where} the proposed method pairs every two WBANs into a cluster and uses cooperative scheduling between the pairs in each cluster to reduce interference.  Node-level scheduling is considered in \cite{movassaghi2014smart} to increase spatial reuse. However, these methods only work for \ds{fixed topologies in the network}. To better cope with \ds{with the required flexibility in} WBAN \ds{implementations}, in \cite{Cheng2011}\cite{grassi2012b2irs},\cite{7299349}, collision avoidance technique\ds{s}, such as beacon rescheduling, channel sensing and adaptive sleeping are used to improve the overall QoS performance for interfering WBANs. Many energy efficient MAC protocols have been proposed, such as \cite{6883350}, \cite{jamthe2014scheduling}. In B-MAC \cite{polastre2004versatile}, the sender \ds{needs to} broadcast a long preamble to be detected by the right receiver to reduce power consumption, but \ds{this}, however, incurs unnecessary \ds{transmission overhead.}

\subsection{Data \yy{R}ate and Power Control}
Comparing with transmission power control schemes, the resource allocation methods tuning some other parameters, such as transmission rate, the packet size and so on, \ds{have proven to be more effective}. Research interests in this area are emerging and some centralized methods  \cite{musku2010game} have been proposed in recent studies. As for distributed algorithms, game theoretic approaches are widely used. For cellular networks,  \cite{hayajneh2004distributed} game theor\ds{etic} schemes with multiple discrete code rates or modulation schemes are proposed, which is also known as link adaptation \cite{ginde2008game}. The mobile terminals updating power and rate by optimizing the Utility Function to Nash Equilibrium. This idea is further extended for wireless Ad hoc network in \cite{6178441}, a simple utility function only depending on Signal to Interference and Noise Ratio (SINR) and the pricing is used. However, for WBANs, there has been limited literature on this subject. \cite{babaei2010transmission} proposed transmission rate adaptation policy for WBANs to improve the QoS by solving a convex optimization problem, but only dynamic postures in WBANs channels are considered.

\subsection{Overall IEEE 802.15.6 Requirements}\label{802.15.6}
As per the IEEE 802.15.6 standard for WBANs its requirements are as follows \cite{ieee15.6}:
\begin{itemize}
\item Bit rates in the range of \SI[per-mode=symbol,per-symbol = p]{10}{\kilo\byte\per\second} to  \SI[per-mode=symbol,per-symbol = p]{10}{\mega\byte\per\second} should be supported via the WBAN links.
\item Packet Delivery Rate ($\mathrm{PDR}$) should be larger than 90\% for a 256 octet payload for more than 95\% of the best-performing links;
\item Up to 256 nodes should be supported by each WBAN;
\item Reliability, jitter and latency should be supported for specific WBAN applications. For instance, medical applications and non-medical of WBANs require latency to be less than \SI{125}{\milli\second} and less than \SI{250}{\milli\second}, respectively; whilst jitter should be less than \SI{50}{\milli\second};
\item WBAN nodes should allow reliable communication in case of mobility scenarios for both on-body and in-body communications;
\item Up to 10 randomly distributed co-located WBAN networks should be supported in a $6\times 6 \si{\metre\squared}$ area.
\end{itemize}
\section{Overall of the Proposed MAC Scheme}\label{Sec:3}

\begin{table}[!htb]
\centering
\caption{Table of Notations }
\begin{tabular}{|c|c|}
\hline
   Symbols & Description\\ \hline
 $N$ & Number of WBANs  		  \\ \hline
 $T_s$ & Length of the superframe   \\ \hline
 $T_b$ & Length of the beacon \\ \hline
 $T_{idle}$ & Length of the inactive period  \\ \hline
 $T_{payload}$ & Length of the transmission period (payload)  \\ \hline
 $T_{ack}$ & Length of the ACK   \\ \hline
 $\yy{T_{\text{slot}}}$ & Length of a time slot  \\ \hline

 $w$ & Backoff counter   \\ \hline
 $m$ & Maximum Backoff Stage   \\ \hline
 $\lambda$ & Persistence coefficient   \\ \hline
 $b$ & Backoff Stage   \\ \hline

 $W$ & Contention window size   \\ \hline
 $CW_{min}$ & Minimum backoff length   \\ \hline
 $CW_{max}$ & Maximum backoff length   \\ \hline

 $p_{f}$ & Probability of transmission failure   \\ \hline
 $b_{i,j}$ & Stationary distribution   \\ \hline
 $\tau$ & Probability of transmission in a randomly slot   \\ \hline
 $p_{c}$ & Probability of packets collision   \\ \hline
 $R_{\tau}$ & Normalization coefficient of $\tau$   \\ \hline

 $\gamma$ & Signal to Noise and Interference ratio(SINR)   \\ \hline
 $h_{ii}$ & Intra-WBAN Channel gain   \\ \hline
 $h_{ij}$ & Inter-WBAN Channel gain   \\ \hline
 $\sigma$ & Noise gain   \\ \hline

 $P_{i}$ & Transmit power of WBAN $i$  \\ \hline
 $R_{i}$ & Data rate of WBAN $i$  \\ \hline

 $\Omega$ & Social Welfare \\ \hline

\end{tabular}

\label{tab:BPSKTable}
\end{table}


Here, we propose a TDMA-based MAC protocol to achieve energy efficiency, reduce inter-\yy{W}BAN interference and obtain low delivery latency for co-existing WBANs. Depending on the special random back-off mechanism to minimize the probability of packets collision among sensors in different WBANs, the overall \ds{level of interference} and power consumption can be consequently reduced.

\subsection{System Model}

We consider star-topology WBANs that \ds{are closely located and coexisting}. Each WBAN consists of a single hub with multiple sensors as described in Fig \ref{fig:sys}. The hub \ds{is not energy-constrained}, so the \ds{energy} consumption of the sensor \ds{is mainly considered}. The hub \ds{facilitates the main WBAN operations} such as synchronization, re-transmission and determining \ds{transmission schedules}.


\begin{figure}[htb!]
  \includegraphics[width=\textwidth,height=8cm]{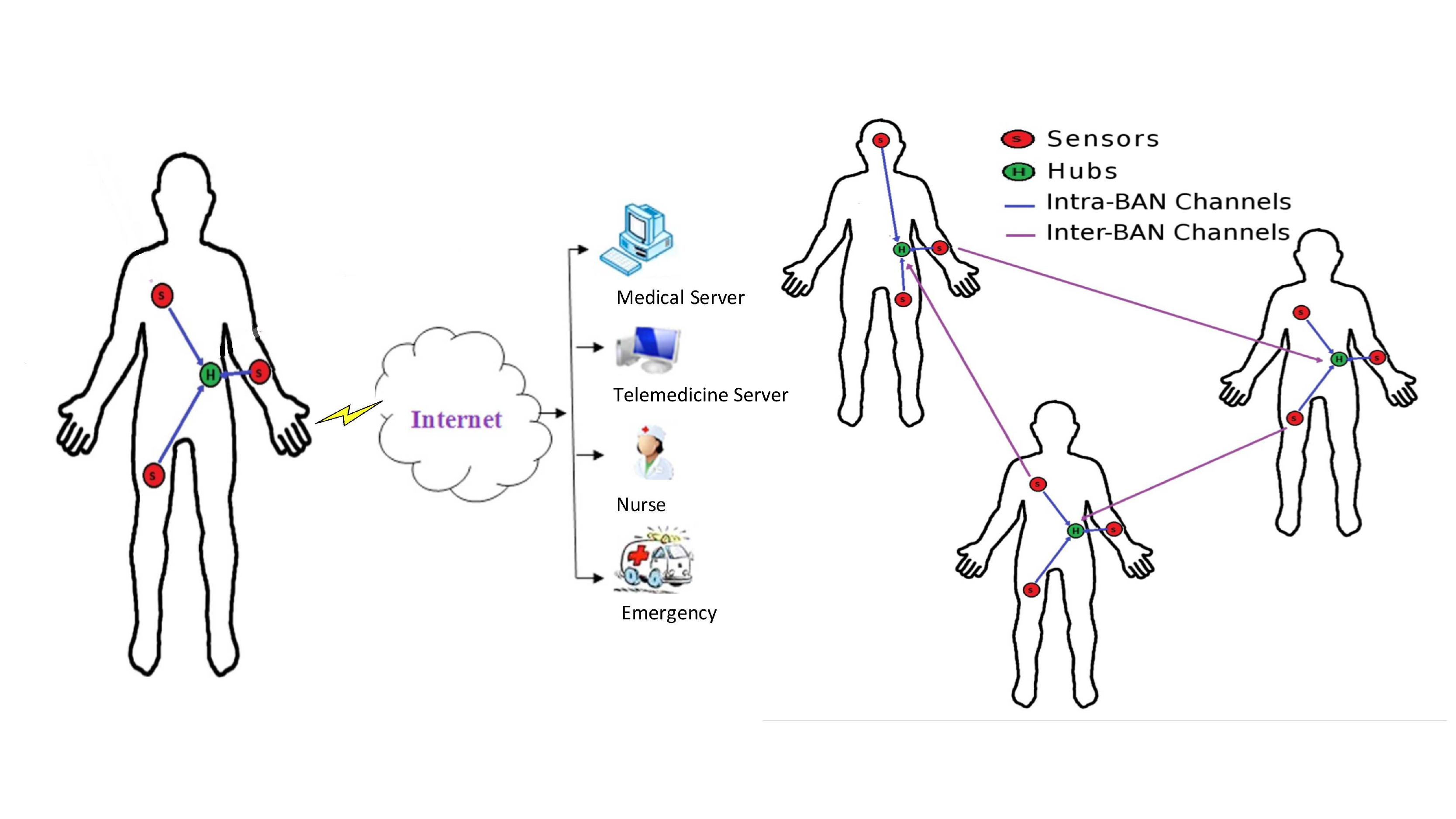}
  \caption{System Model}
  \label{fig:sys}
\end{figure}

As interference \ds{across} all co-existing WBANs \ds{is the main reason for dropped-packets and energy wastage, the following assumptions apply}:

A1. All BANs are \ds{saturated}, which means that they always have a packet to send (packer arrive rate is 1).

A2. Interference-dominated, where packet loss due to \ds{additive noise is} negligible.

A3. \ds{All hubs are not energy constrained}

A4.  \ds{T}ime-slots \ds{across every} superframe \ds{are normalized to unity}

Within each WBAN the sensor devices acquire data and \ds{use TDMA to transmit to the hub} to avoid idle-listening and overhearing. The intra-BAN interference is collision free when each sensor is transmitting using round-robin scheduling. \ds{However}, it is infeasible to implement central coordinator among co-existing WBANs \ds{that are closely located}. Thus, the transmission between different WBAN\ds{s} can not be synchronized. Therefore, co-channel interference may arise due to collisions among\ds{st} concurrent transmissions made by sensors in different WBANs.

\subsection{\ds{MAC} Layer Specification}

\subsubsection{Superframe Structure}
As described above, the sensor\ds{s} in each WBAN are synchronized by periodic transmission of the superframe, with constant length $T_s$. Each superframe consists of a beacon, ACK reception, several slots used for data transmission and \ds{an} inactive (idle) period. The beacon \ds{is} configured \ds{to contain} control information for broadcasting to sensors. The turnaround time is negligible. After the reserved transmission slots followed by ACK reception, the remaining slots in the superframe are considered as \ds{an} inactive or idle period. The basic structure of the superframe is depicted in Figure \ref{fig:MAC}.

Each superframe starts with broadcasting a \ds{b}eacon packet from the hub to the sensors, which consists of information for establishing links and synchronization. Due to Assumption A.3, the probability that beacon or ACK \ds{is not received at the sensor} is negligible, as the hub can apply higher transmit power to avoid SINR outage. The beacon is immediately followed a data frame, which mainly consist of up-link data traffic from the sensors to the hub.

\begin{figure}
\centering
\includegraphics[trim=6cm 6cm 6cm 6cm, clip=true,width=1\columnwidth]{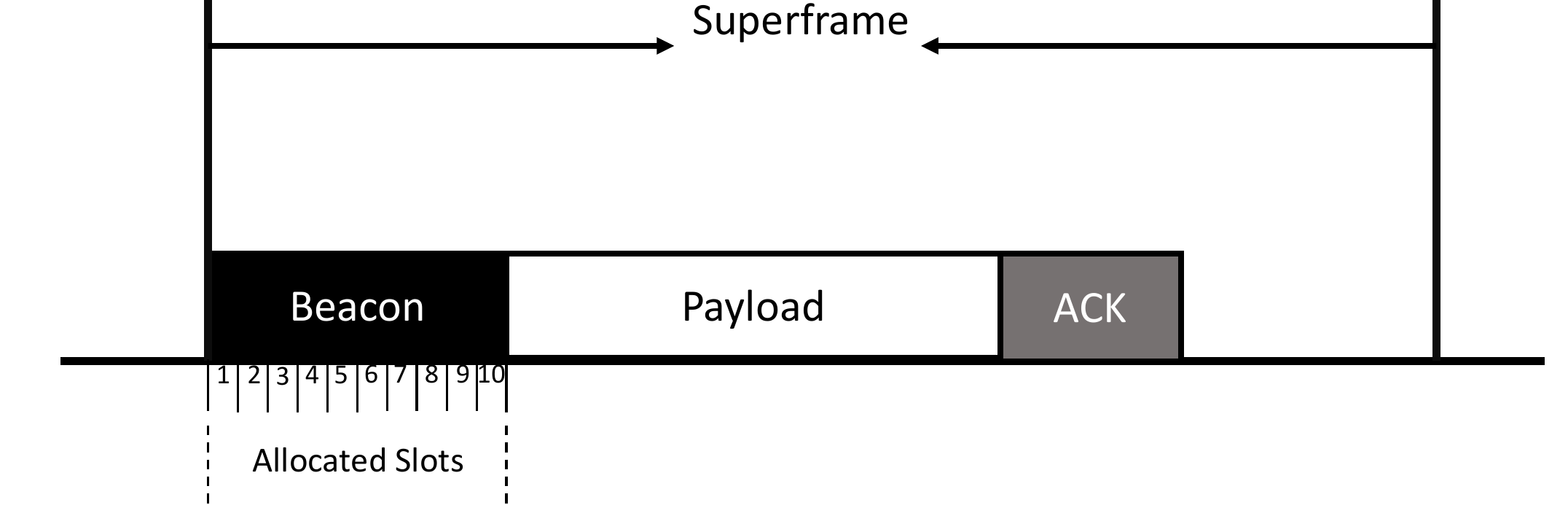}
\caption{Superframe Structure}
 \label{fig:MAC}
\end{figure}

\subsubsection{Back-off Mechanism}

In order to avoid collisions when many WBANs sharing the medium, \ds{a} back-off algorithm can \ds{be} used in WBAN networks. Since, if a transmission \ds{fails due to} interference, it is highly likely that \ds{a following transmission will also be} blocked. Keeping the sensor in back-off states for a shorter period of time will increase packet  delivery ratio and reduce collision probability.

Fig.\ref{fig:Time} illustrate\ds{s} the back-off mechanism and the operation mode of the sensor. As depicted in Fig.\ref{fig:sys}, all the sensors may be in one of three different operation states: back-off state, transmission state, and inactive state.

\begin{figure}[!htb]
\centering
\includegraphics[trim=0cm 4cm 1cm 4cm,clip=true,width=1\columnwidth,height=6cm]{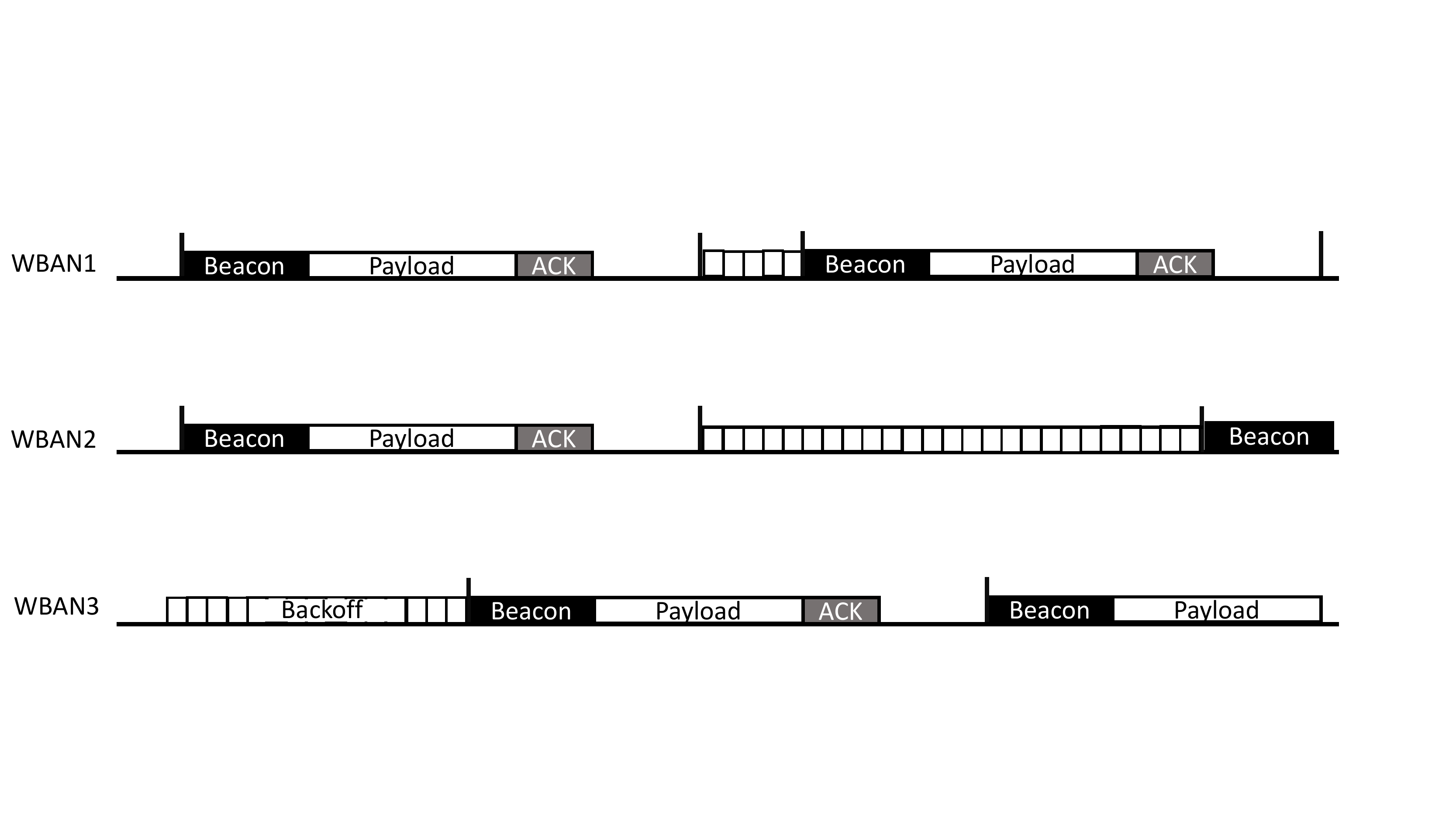}
\caption{Timing Scheme and Operations}
 \label{fig:Time}
\end{figure}

If the ACK is not received by the sensor till the end of a superframe, the sensor executes a random backoff procedure before the next beacon is transmitted by the hub to reduce the collision probabilities among WBANs.

\yy{During the random backoff procedure, each sensor set its backoff counter $w$ as \ds{a} random integer uniformly distributed over the interval of $w\in[1,W]$.} The value $W$, \yy{namely Contention Window size is determined} by the Backoff stage $b$ and the number of Maximum Backoff Stages $m$. \yy{The Backoff stage $b$ equals to} the number of transmissions failed for the packet. At the first transmission attempt,  $W$ is set \yy{as} 0.  After each unsuccessful transmission,  $W$ is set as the $CW_{min}$ multiplied by the persistence coefficient $\lambda$ $(W = \lambda^b CW_{min} )$, until it reach the maximum value $CW_{max} = \lambda^m CW_{min} $ (in this paper $\lambda=2$). The value $CW_{min}$ called minimum back-off length.
The basic algorithms are listed in Algorithm \ref{alg:1}\yy{.}

\begin{algorithm}
\caption{The Proposed MAC Layer Protocol}\label{euclid}
\begin{algorithmic}[1]

\State Initializing MAC Parameters when, $W = CW_{min}$ , $b=0$ ,$w= 0$
\State Hub send the sensor a beacon to sensor $i$ with synchronization information  \label{marker}

\State After receive the beacon, the sensor $i$ transmit data using the scheduled slots.

\If {Hub successfully receives the packet }
\EndIf
\State Sensor keeps inactive for a period of $T_{idle}$ until the end of the superframe.
\State \Goto{marker}
\State $i \gets i+1$

\State \textbf{close};

\If {The transmission fails because of the interference, the sensor will not receive the ACK}
\State $b \gets b+1$
\If {$b > MaxRetransmissionLimit$}
\State  Discard The Packet
\State \Goto{marker}
\EndIf

\State \textbf{close};
\State The hub calculate\ds{s} the backoff length
\State $W_b = \lambda^b CW_{min}$
\State $w= random\{0, W_b\}$

\State Sensor keeps inactive for a period of $T_{idle}$, and start\ds{s} back off until the end of the superframe.

\State $w \gets w-1$
\If {$w = 0$}
\State \Goto{marker}

\State \textbf{close};

\EndIf
\EndIf

\end{algorithmic}
\label{alg:1}
\end{algorithm}


\subsection{Markov Model}

\ds{A} Markov model \ds{was} initially proposed by Bianchi for IEEE 802.11 DCF \cite{840210}. The model describes the basic fundamental process of MAC layer scheme through a Markov chain. The model has been extended in several directions.

In this section, we proposed a discrete-time Markov chain, which models the operation of the proposed algorithm in the tagged WBAN and captures the key characteristics of the MAC layer timing scheme, such as, superframe structure, and re-transmission \yy{mechanism}. The Markov chain model can help us investigate \ds{features} of the proposed MAC layer timing scheme, such as throughput \yy{and} delay.

Figure \ref{Markov}. shows the state transition diagram of the Markov chain of the proposed MAC scheme. In the Markov model, the state at time $t$ for tagged WBAN is represented by the stochastic process $( b(t),w(t) )$.  $b(t),w(t)$ represent the back-off stage and the back-off counter respectively. Here, $b(t) \in [0, m]$ represents the back-off stage of the tagged WBAN at time $t$, where $m$ is the maximum backoffs. $w(t)$ represents the value of the backoff counter at time $t$.

\begin{figure}
\centering

\includegraphics[width=9cm,height=9cm]{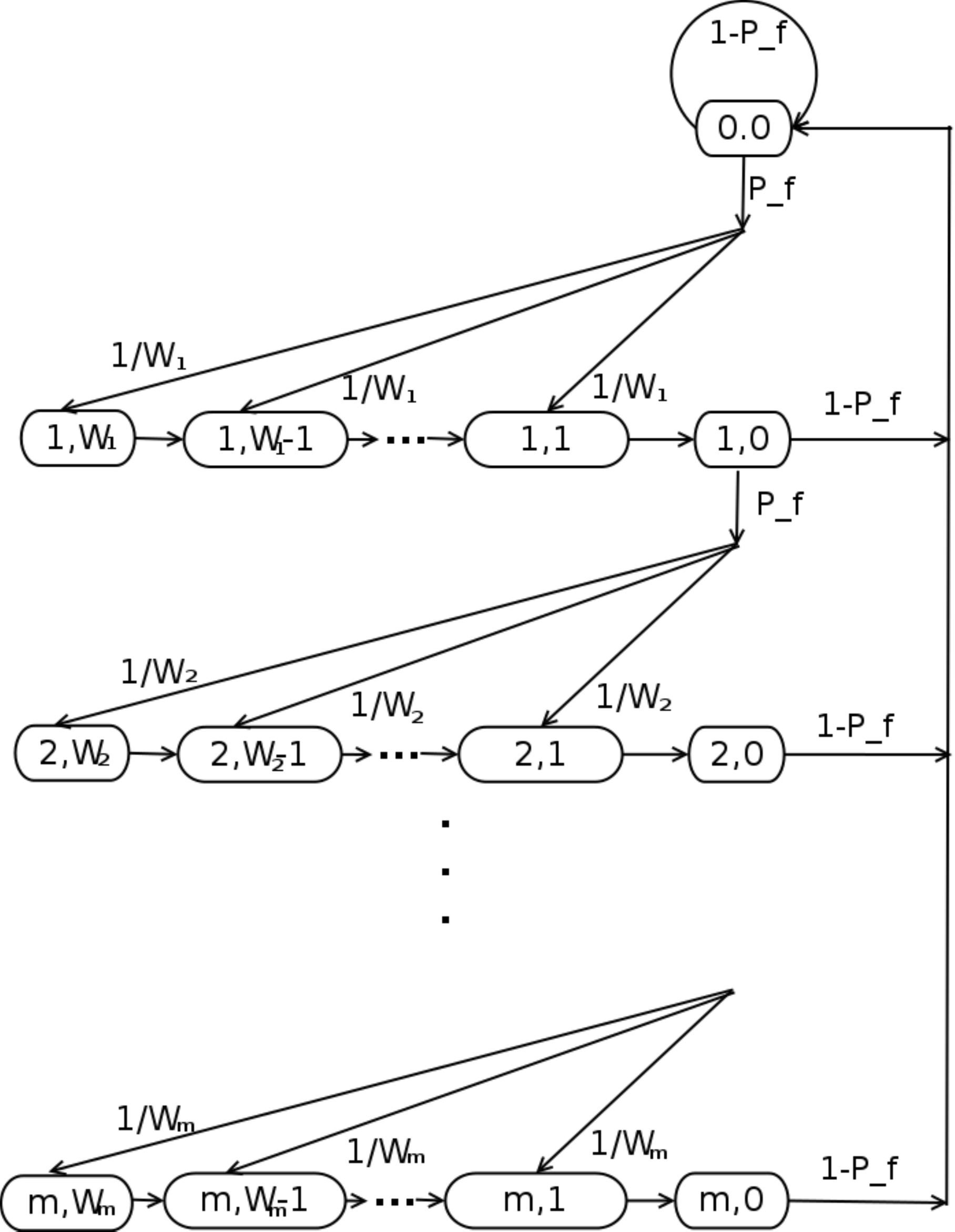}
  \caption{Markov Chain for Proposed MAC Scheme} 
  \label{Markov}
\end{figure}

Initially, the value\ds{s} of back-off counter and re-transmission counter are set to zero. In the back-off stage, when the value of back-off counter reache\ds{s} zero, the hub will re-transmit the packet. If the transmission attempt is successful, the state will move to $\{0,0\},$ \ds{($b$($t$)= 0, $w$($t$) = 0)}, and the WBAN starts to transmit \ds{in the} next frame after the in-active period. Otherwise, the sensor moves from back-off stage $j$ to $j+1$ with its re-transmission counter increment\ds{ed} by 1. In addition, at the state $\{m,0\}$, the data from the sensor will either be successfully transmitted or discard\ds{ed} by the sensor. The parameter $p_f$ represents the probability of transmission failure due to the inter-BAN interference, which denotes the probability of ACK reception under collision \ds{of packets}. The \ds{state transition probabilities are}:

\begin{equation}
\begin{aligned}
\begin{cases}
P_r((i,j-1)|(i,j)) & = 1\\
P_r((i+1,j)|(i,\yy{0})) & = p_f \frac{1}{W_i}\\
P_r((0,0)|(i,0)) & = 1 - p_f \\
P_r((0,0)|(0,0)) & = 1 - p_f
\end{cases} ,
\end{aligned} 
\end{equation}
where $W_i = 2^{i} CW_{min}, i \geq 1$.

\subsubsection{\yy{Steady-State Solution}}
To complete the construction of the Markov chain in Fig. 1, \yy{a form solution} for the Markov chain \yy{is required}.

Let the stationary distribution of the Markov chain be, $b_{i,j}=\lim_{t\to\infty} Pr(s(t)=i,b(t)=j)$, where $ i\in (0,m), j\in(0,W_i)$. Then, we can calculate the stationary distribution for all the values of $b_{i,j}$. By using the equation (1),  $b_{i,j}$ can be expressed as functions of the value  $b_{0,0}$ and of the transmission failure probability $p_f$. $b_{0,0}$ is finally determined by imposing the normalization condition:

\begin{equation}
\begin{aligned}
b_{0,0} + \sum_{k=1}^{m} \sum_{j=0}^{W_k} b_{k,j} = 1
\end{aligned}
 \label{equation:Pro}
\end{equation}

that simplifies \ds{to}:

\begin{equation}
\begin{aligned}
b_{0,0} = \frac{2(1-p_f)(1-2p_f)}{2CW_{min}p_f(1-(2p_f)^m)(1-p_f)+(2+p_f)(1-p_f^m)(1-2p_f)}
\end{aligned}
 \label{equation:Pro}
\end{equation}

We can now express the probability  $\tau$ that a WBAN sensor attempts to carry out a superframe transmission in a randomly chosen slot time. The transmission occurs when the backoff time counter is equal to zero. Thus, we can write the probability $\tau$:

\begin{align}
 \tau= &\sum_{k=0}^{m} b_{k,0}\\
     = & \frac{2(1-2p_f)(1-p_f^{m+1})}{2 CW_{min} p_f(2-(2p_f)^{m})(1-p_f) + (1-2p_f)(2+p_f)(1-p_f^m)}
\label{tau}
\end{align}


With probability $\tau$, the probability of transmission failure $p_f$ is the intersection of the probability of collision and the packet error rate ($\mathrm{PER}$) :

\begin{equation}
\begin{aligned}
p_f = p_{c} \cap \mathrm{PER} = (1-\mathrm{PDR})(1-(1-\yy{\eta}_{\tau}\tau)^{N-1}),
\end{aligned}
 \label{equation:Pro}
\end{equation}
where  $\mathrm{PDR}$ is the packet deliver ratio and $\yy{\eta}_{\tau}$ is a coefficient that normalize\ds{s} $\tau$ by the average number of time slots staying in an arbitrary state of the Markov chain.

By using (\ref{tau}) and (\ref{equation:Pro}), the value\ds{s} of $p_f$,$\tau$ and $b_{0,0}$ can be solved and used in  performance metrics.

\subsection{Performance Validation}
By comparing the analytical model\ds{ing} with simulation results, we can evaluate the accuracy of the Markov Model described in the previous section. We simulated $N$ \yy{single-link star topology} WBANs co-existing in the saturated regime (all WBANs always have a packet to transmit). A performance analysis is also conducted based on the steady-state solutions. We assumed when intra-WBAN interference occurs \ds{that} the packets collision will result in transmission failure $p_f = p_{c}$ as $\mathrm{PER} = 1(\mathrm{PDR}=0)$.

In \yy{our experiment, the data rate is set as \SI{32}{\kilo\bit\per\second}}. The retransmissions limits \yy{$m=4$}, and the \yy{minimum back-off length} $CW_{min} = \yy{64}$. \yy{The duration of one} superframe is $T_{s} = 64ms$. \yy{Each superframe consist of 256 time slots.} \yy{Each time slot last for $T_{\text{slot}} = 250\mu s$}. The beacon, payload and the ACK are set to be 30,110,10 bytes respectively.

\begin{figure}[!htb]
\centering
\includegraphics[width=8cm,height=6cm]{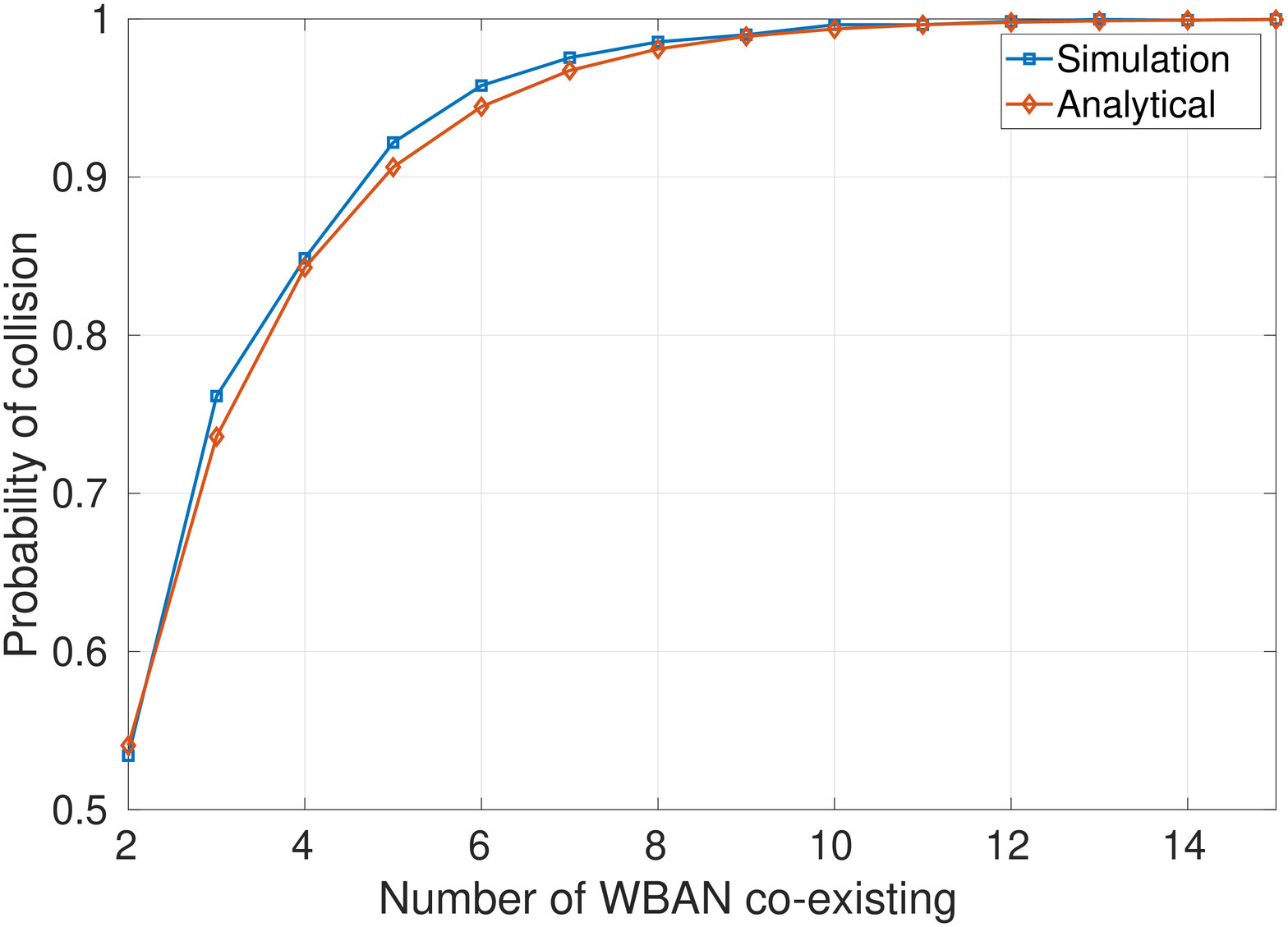}
\caption{The Comparison of Probability of Collision}
 \label{fig:CP_vali}
\end{figure}

\begin{figure}[!htb]
\centering
\includegraphics[width=8cm,height=6cm]{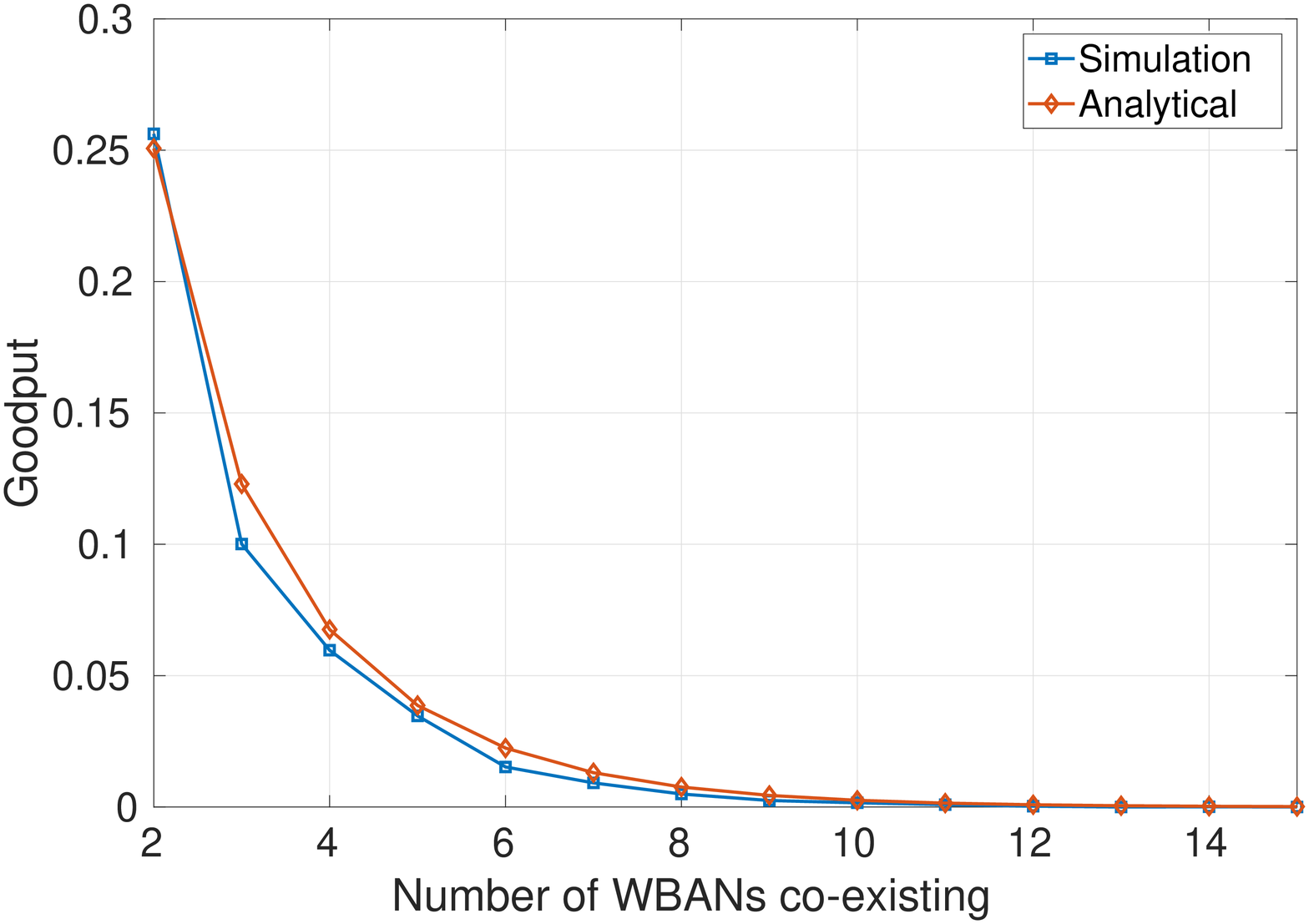}
\caption{The Comparison of \yy{Goodput}}
 \label{fig:TP_vali}
\end{figure}

Fig. \ref{fig:CP_vali} and Fig. \ref{fig:TP_vali} \ds{show} that, in terms of Probability of Collision and Goodput (S), the analytical model \ds{matches} the simulation results \ds{accurately}. Goodput is defined as the ratio between time \ds{elapsed} \ds{successfully delivering a} packet and the total time. Analytically, the goodput is calculated \ds{as}:
\begin{equation}\label{equation:Goodput}
\begin{aligned}
Goodput (S) = \frac{(1-p_{f}) \tau T_{payload}}{\tau  T_{s} + (1-\tau) \yy{T_{\text{slot}}}},
\end{aligned}
\end{equation}
where
\begin{equation}
\begin{aligned}
T_{s} = T_{beacon} + T_{payload} + T_{ack} + T_{idle}
\end{aligned}
 \label{Tsf}
\end{equation}
$T_{\yy{beacon}},T_{payload},T_{ack},T_{idle}$ represent the duration \ds{of} beacon, payload, ACK and idle respectively. However, weak matching occurs when less WBANs are co-existing, as randomness \yy{of the back-off counter in} the simulation is relatively large. Fig. \ref{fig:TP_vali} shows that when more WBANs \ds{are} co-located, the system's \yy{Goodput} decreases as the collision probability increases.

\section{Link Adaptation Game Algorithm}\label{Sec:4}
The majority of existing game theoretic algorithm\ds{s} in WBANs only \ds{focus} on transmit power control. However, \ds{when} considering saturated traffic conditions, in particular interfering networks, selfishly changing the transmit power may increases packet losses, and thus, reduce the overall throughput. Thus, we observed that greater system efficiency \ds{can be achieved} by varying performance-impacting characteristic such as modulation or data rate.

One of the objectives in th\ds{is} paper is to exploit the above MAC layer scheme to find a method that can further increases the system's \ds{energy} efficiency and reduce inference. In this section, we develop a utility-based game theory model that is a function of two variables: transmi\ds{ssion} power and data rate to address interference and packet contentions. The transmitter chooses the value of the transmit power and data rate to maximize \ds{energy} efficiency whilst meeting the packet delivery ratio ($\mathrm{PDR}$) requirements. Since the transmitter's action will be a function of the choice of data rate \ds{and transmit} power. Therefore, the two parameters need to \ds{be executed} jointly.

According to experimental results, \ds{an increase in} data rate can result in WBAN's packet delivery ratio degradation \ds{in} the same SINR regime. However, the length of payload transmission time is also related to the WBAN's data rate, \ds{so that} higher data rate may also reduce the intra-WBAN packet collision possibility as the transmission time is reduced accordingly, and thus mitigate interference. Furthermore, the proposed game theor\ds{etic} algorithm handles packet re-transmissions, until the retry limit \ds{is reached}. In addition, it should be noted that all game associated calculations are made by the hub, hence there will be no extra computational cost for the sensors.

\subsection{Game-theoretic System \yy{M}odel}

We consider the system model of the form described in Section 2.1 where multiple WBANs are co-located. All WBANs are within each others \ds{interference range}, and the corresponding transmitters always have packets to send. TDMA schemes with \ds{the} same random back-off scheduling mechanism are applied for all WBANs and the minimum back off slot length cannot be changed. There are no pre-assigned priorities among different WBANs \ds{such} that all links can expect identical priorities in traffic. Formally, we define the Link Adaptation Game \ds{in normal} form as $G = \{\mathbf{N},\mathbf{(P,R)},U\}$, \ds{notation described in the following sub-sections.}

\subsection{Players}

Each co-existing WBAN is a player in the \yy{Link Adaptation Game}, the player set is denoted by $\bf{N} $ = $\{1,2,3...N\}$. Within each WBAN, the sensor(transmitter) adapts its power and data rate \yy{by utilizing} the game-theoretic algorithms.

\subsection{Action Set}

We require that the data rate of WBAN $i$, ${R_i}$ is chosen from a discrete finite set $\overline{R} = \{R_{min} ,...,R_{max} \}$, where $R_{min}$ is the base rate and $R_{max}$ is \ds{the} maximum data rate. \yy{In each packet transmission}, \ds{and} variations \ds{in ${R_i}$} \yy{result in \ds{variations in} packet delivery ratio (PDR)}. In addition, each WBAN can adjust its power $P_i$ $\in$ $[P_{min},P_{max}]$. Thus the action selected by any player $i$ is defined as the pair $A_i = (P_i,R_i)$, where $P_i$ is the transmit power of player $i$ and $R_i$  is the data rate of player $i$. The aggregation of all players action is denoted as $\bf{A} = \yy{\bf{(P,R)} =  \{\bf{A_1,A_2...A_n}\} }$.

\subsection{Utility Function}

It is obvious that each player is always trying to maximizing its own utility.  However, due to the non-cooperative nature of this game, it is easy to see that in an attempt to maximize its own benefits at any cost, each WBAN is likely to consume maximum power, and the \yy{highest} data rate. This will also create excessive interference, leading to performance degradation. \ds{A} pricing mechanism is \ds{then} also introduced to penalize the use of excessive power. The utility function of each player is defined as follows:

\begin{equation}
U(P,R) = -C(P)+ \ln(1 + \mathrm{PDR}(P,R)) + G(R),
\label{Utility_Rate}
\end{equation}
where,
\begin{equation}
\begin{aligned}
C(P)= c \cdot P^g \\
G(R)= - q \cdot \frac{1}{R}
\end{aligned}
\end{equation}

the coefficients in the cost function $c,g,q$ are positive constants that can be tuned depend\ds{ing} on channel conditions. The linear cost function \yy{(where the exponent $g=1$)} \ds{is} commonly used in the literature, however, in the proposed game $g$ are generally greater than 1 to provide strictly concavity. $c,q$  \ds{are} constant non-negative weighting factors.

In equation \eqref{Utility_Rate},  \yy{a sigmoid approximation of $\mathrm{PDR}$ is} introduced. This model has been shown to be capable of approximating the PDR versus SINR of a wireless channel. Here, the SINR is defined as:
\begin{equation}
\gamma_{i} = \frac{h_{ii}P_{i}}{ \sum_{j = 1, j \neq i}^{N} h_{ij}P_{j} + \sigma^2 },
\end{equation}
where, $h_{ij}$ represents the channel attenuation between WBAN $i$ and WBAN $j$. \yy{$h_{ii}$} denotes the on-body channel attenuation in WBAN i. $\sigma$ is the noise \yy{gain}. The sigmoid $\mathrm{PDR}$ function of SINR is presented as follow\ds{s}:
\begin{equation}
\mathrm{PDR} = \exp \left( \alpha \cdot \gamma ^\beta \right)
\end{equation}

Fig. \ref{fig:PDR_est.eps} shows the comparison between approximated and simulated $\mathrm{PDR}$ vs. SINR for different data rate ($R_1 \sim R_5$) respectively. The sets of parameters $\yy{\alpha,\beta}$ of the sigmoid model in the equation above that best approximate the simulation curves are selected by computer-aided search and summarized in Table \ref{tab:BPSKTable}:
\begin{figure}[!htb]
\centering
\includegraphics[width=8cm,height=6cm]{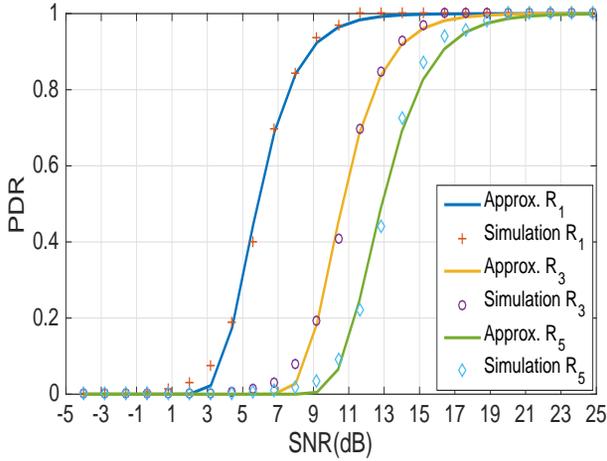}
\caption{PDR vs. SINR}
 \label{fig:PDR_est.eps}
\end{figure}

\begin{table}[!htb]
\centering
\caption{Coefficients for PDR Estimation }
\begin{tabular}{|c|c|c|}
\hline
   Data Rate & $\alpha$&$\beta$\\ \hline
 $R_1= 25.6 kbps$ &-100.02  &-3.66		  \\ \hline
 $R_2= 51.2 kbps$ &-214.95  & -2.82 \\ \hline
 $R_3= 76.8 kbps$ & -663.69 & -2.79 \\ \hline
 $R_4= 102.4 kbps$ & -1182.7 & -2.73 \\ \hline
 $R_5= 128.0 kbps$ & -1433.5 & -2.58 \\ \hline
\end{tabular}
\label{tab:BPSKTable}
\end{table}
\subsection{Nash Equilibrium}

The Nash equilibrium is a set of strategies that \ds{guarantee the best response of each player with respect to the chosen utility}. In the proposed non-cooperative game, the game is played by rational players, which implies that every player \ds{adopts the strategy achieving the} Nash Equilibrium.

\theoremstyle{definition}
\begin{definition}

Let $\bf{A^\ast} = \bf{(P^{\ast},R^{\ast})}$ be the Nash Equilibrium in the Link Adaptation Game, then for every $i\in N$:

\begin{equation}\label{eq:PR_ult}
U(P_i^{\ast},R_i^{\ast}) \geq U\yy{\{}(P_i,R_i),\bf{P^{\ast}\boldsymbol{_{-i}},R^{\ast}\boldsymbol{_{-i}}}\yy{\}},
\end{equation}
\end{definition}
where, $(\bf{P_{-i},R_{-i}})$ represents all other player strategies except for player $i$. At the end of every transmission (say $t$), players update their next transmit power and rate jointly to maximize the outcome \yy{of adapting} the utility function based on the current SINR:

\begin{equation}\label{eq:pr_ult}
\{P_i(t+1),R_i(t+1)\} = \argmax U \{ (P_i(t),R_i(t)),\mathbf{P^{\ast}\boldsymbol{_{-i}},R^{\ast}\boldsymbol{_{-i}}} \}
\end{equation}

The action profile $\mathbf{P} ^\ast=(P_1^\ast,P_2^\ast,P_3^\ast...P_n^\ast) $ for $n\geq2$ is a Nash Equilibrium, is the best response towards $\mathbf{P}^\ast\boldsymbol{_{-i}}$.

\subsection{Existence and uniqueness of the Nash equilibrium}

The existence and uniqueness of the Nash Equilibrium of the
proposed game are proved as follows:

\begin{lemma}\label{lemma1}

The action space $\overline{A} = (\overline{P},\overline{R})$ is not a convex set. However, under the condition that $R$ is fixed, $\overline{A} = (\overline{P},\overline{R})$  is a convex set.
\end{lemma}

\begin{proof}
We simply choose two points, $ A_1 = ({P_{max},R_1})$, $A_2 = ({P_{max},R_2})$ that have the same power component. A line connecting these points consists of only two points $A_1$,$A_2$ themselves. The intervening points on this line do not belong to $(P,R)$. Hence, $(P,R)$ is not convex.
However, suppose the data rate is fixed at, say $R_1$ . We note that a convex combination
$ A^\prime = \Lambda A^\prime_1 +(1-\Lambda) A^\prime_2$, $\Lambda \in [0,1]$, where $A^\prime_1 = (P_1,R_1),A^\prime_2= (P_2,R_1)$ are any two arbitrarily selected actions in $A$, such that $P_{min}< P_1, P_2 < P_{max}$  belongs to the set $\overline{A}$. Hence the set $\overline{A}$ is convex when $R$ is fixed.
\end{proof}

 \begin{theorem}\label{theorem1}
 (Existence) The game $G$ admits at least one Nash Equilibrium, when assuming $R$ is fixed.
 \end{theorem}
\begin{proof}
In game $G$, for \yy{WBAN $i$}, $i \in \mathcal{N}$ the following condition can be verified. When the \yy{data rate of WBAN $i$} is fixed, the action set $A = (P,R)$ is a nonempty, convex, bounded in finite dimension vector space as proved in \textbf{Lemma} \ref{lemma1}. The utility function $U$ is continuous for all $P_i\in[P^{min}_i,P^{max}_i]$. As the first derivative of the utility function $U$ is well defined as:
 \begin{equation}
 \frac{\delta U_i}{\delta P_i}=-c \cdot gP^{g-1}+\left(1-\frac{1}{(1+\mathrm{PDR}_i)}\right)\alpha\beta\frac{\gamma_i^{\beta}}{P_i},
 \end{equation}
where $\frac{|h_i^i(k_i)|^2}{I_-i}=\frac{\gamma_i}{P_i}$, therefore, as $P_i\in[P^{min}_i,P^{max}_i]$ is real and the $\mathrm{PDR}_i$ is non-zero, the Theorem \ref{theorem1} is proved.
\end{proof}
 \begin{theorem} \label{theorem2}
 (Uniqueness): The Nash Equilibrium in each stage of the game $G$ is unique, and independent of history so it is a unique sub-game perfect equilibrium.
 \end{theorem}
\begin{proof}
The second derivative of $U(\cdot)$ is shown to be always negative $\forall{i}$, so that $U(\cdot)$ is strictly concave.
\begin{multline}
\frac{\delta^2 U_i}{\delta P_i^2}=-c \cdot g(g-1)P_i^{g-2}
+\left(1-\frac{1}{(1+\mathrm{PDR}_i)}\right) \gamma_i^\beta \alpha \beta(\beta-1)/P_i^2\\
-\frac{\mathrm{PDR}_i}{(1+\mathrm{PDR}_i)^2}\alpha^2\beta^2\gamma_i^{2\beta}/P_i^2,
\end{multline}
where $\mathrm{PDR}_i$ is always positive and between (0,1), and the term $(\beta-1)$ is less than 0. In addition $w$ is positive, thus $\frac{\delta^2 U_i}{\delta P_i^2}<0$. Therefore the utility function has a global maximum at $P_i^\ast$ which occurs at the point where $\frac{\delta U}{\delta P_i}=0$.
\end{proof}

\begin{algorithm}
\caption{Main steps for Link Adaptation Game}\label{Algorithm:2}
\begin{algorithmic}[1]

\State Initializing MAC Parameters when, $CW = CW_{min}$ , $b=0$ ,$w= 0$
\State Hub sends a beacon to sensor $i$ with synchronization information  \label{marker}

\State After receiving the beacon, the sensor $i$ transmit\ds{s} data using the scheduled slots with $(P(t),R(t))$.

\If {Hub successfully receives the packet }

\State The Hub choosing the transmit power and rate for  next transmission by equation $\eqref{eq:pr_ult}:\{P_i(t+1),R_i(t+1)\} = \argmax U\{(P_i,R_i),\bf{P_{-i}^{\ast},R_{-i}^{\ast}\}} $

\State Sensor keeps inactive for a period of $T_{idle}$ until the end of the superframe.
\State \Goto{marker}
\State $i \gets i+1$
\State $t \gets t+1$
\State \textbf{close};
\ElsIf {The transmission fails because of the interference, the sensor will not receive the ACK}
\State $b \gets b+1$
\If {$b > MaxRetransmissionLimit$}
\State  Discard The Packet
\State \Goto{marker}
\State \textbf{close};
\EndIf

\State The Hub choosing the transmit power and rate for  next transmission by equation $\eqref{eq:pr_ult}:\{P_i(t+1),R_i(t+1)\} = \argmax U\{(P_i,R_i),\bf{P_{-i}^{\ast},R_{-i}^{\ast}\}} $
\State The hub calculate\ds{s} the backoff length \label{opretion:game Start}
\State $W = \lambda^b CW_{min}$
\State $w= random\{0, CW\}$
\State Sensor keeps inactive for a period of $T_{idle}$, and start\ds{s} back off until the end of the superframe.\label{opretion:game end}

\State $w \gets w-1$
\If {$w = 0$}
\State \Goto{marker}
\State $t \gets t+1$
\State \textbf{close};

\EndIf
\EndIf

\end{algorithmic}
\end{algorithm}

However, in \ds{practice}, \ds{the} data rate $R$ is not always fixed. Hence, we need to make sure that the game only admits a unique Nash Equilibrium solution over the action space $\overline{A}$.  As the concavity of the utility function leads to the uniqueness of the Nash Equilibrium, we \ds{use} the concept of potential game\cite{MONDERER1996124}, which  provides useful properties concerning the justification of Nash equilibrium.

\subsection{Forming a Potential Game}

For game $G$, when at high $\mathrm{PDR}$ regime (where the system usually operates), we can obtain the following approximation by using Taylor expansion for the last term in the utility function (\ref{eq:PR_ult}):

\begin{equation}
\ln(1+\mathrm{PDR}(A_i,A_{-i})) = \ln(1 + \exp (\alpha \cdot \gamma ^{\beta} )) \\ \approx \ln(2) + \frac{\alpha \cdot \gamma  ^{\beta} }{2} + \frac{(\alpha \cdot \gamma  ^{\beta})^2 }{8}  ...
\end{equation}

Substitute this approximation in to the utility function (only takes the first order term), we can get the a new game with utility function defined as:
\begin{equation}
UP_i(P,R) = -c \cdot P_i^g+ q \cdot \frac{1}{R_i} + \ln(2) +  \frac{\alpha(R_i) \cdot \gamma  ^{\beta(R_i)} }{2} \approx U(P,R)
\end{equation}
\ds{Thus we transform the game $G$ to a potential game denoted} by $G_P = \{\mathbf{N}, \mathbf{A}, U_P\}$, where  $\yy{U_P}$ is the new utility function.

We firstly provide the definition of the exact potential game, and proceed to show that the game \ds{belongs to} the class of exact potential game\ds{s}.

\begin{definition}
A game is said to be an exact potential game if there exists a function satisfying:

\begin{equation}
U(S_i,S_{-i}) - U(T_i,S_{-i}) = F(S_i,S_{-i}) - F(T_i,S_{-i}),
\end{equation}
where $F$ is called the potential function that can map the action space of the game in to \ds{a} real space.
\end{definition}

The Game $G_P$ is an exact potential game, with a potential function defined as:

\begin{equation}
F_i(A) = \sum -cP_i^w - q\frac{1}{R_i} +  \frac{\alpha(R_i) \cdot \gamma ^{\beta(R_i)} }{2}
\end{equation}

Notice that we discard $\ln(2)$ in the utility function when constructing the potential function as the constants can be canceled out. We can see that $F$ satisfies the Definition 2.1 in \cite{la_potential_2016},

\begin{multline}
F_i =-cP_i^g - q\frac{1}{R_i} +  \frac{\alpha(R_i) \cdot \gamma  ^{b(R_i)} }{2} +  \sum_{j \neq i} -cP_j^w - q\frac{1}{R_j} + \frac{\alpha(R_j) \cdot \gamma ^{b(R_j)} }{2} \\ \hspace{-75mm} \mathrm{and,} \\
\\ F(A_i,A_-i)-F(T_i,A_{-i})  =  -c P_i^g - q\frac{1}{R_i} + \frac{\alpha(R_i) \cdot \gamma  ^{b(R_i)} }{2}  -c P_{T^g_i} - q\frac{1}{R_{T_i}} \\ + \frac{\alpha(R_{T_i}) \cdot \gamma  ^{b(R_{T_i})} }{2} = U_P(A_i,A_{-i}) - U_P(T_i,A_{-i}),
\end{multline}

where $T_{i} = (P_{T_i},R_{T_i})$, and thus it is a potential function of the game $G_P$. Also, game $G_P$ is a best response potential game, which defined as :

\begin{definition}
The game $G$ is a best-respon\ds{se} game if and only if a potential function $F$ exist\yy{s} such that,
\begin{equation}
\argmax U(A_i,A_{-i}) = \argmax F(A_i,A_{-i})
\end{equation}

\end{definition}
According to \cite{VOORNEVELD2000289}, this leads to the following lemma:



\begin{lemma}
For best-response game $G$ defined over action space $\overline{A}$, with a potential function $F$. If $A\in\overline{A}$ maximizes $F$, the\ds{n} it is a Nash Equilibrium for $G$.
\end{lemma}

\subsection{\yy{Large Midpoint Property} and Discrete Concavity}
For an exact potential game, the change of the potential function attributes the same amount of change in a player's utility function due to its strategy deviation. A concave potential function guarantees that every Nash equilibrium of the game also maximizes a potential function.

Therefore, with the help of the results in \cite{ui_discrete_2008} on discrete concavity for potential games, we can prove the uniqueness of the Nash equilibrium in $G_p$: since the maximizer in the potential function $F$ is unique, so is the Nash Equilibrium in game $G$.

According to  \cite{ui_discrete_2008} the  Large Midpoint Property (LMP) is defined as:

 \begin{definition}
 For a function defined over discrete set satisfies LMP if for any $x,y\in X$ with $|x-y|=2$ ,
 \begin{equation}
 \max_{|x-z|=|z-y| = 1} f(z) = \left\{
\begin{aligned}
>& \min\{f(x),f(y)\}  ,\mathrm{if}  f(x)\neq f(y) \\
\geq & f(x) = f(y) \mathrm{,otherwise}
\end{aligned}
\right.
 \end{equation}
 \end{definition}

We show that the potential function satisfies the LMP the discrete strategy $\overline{R}$. As the choice of data rate $[R_{\yy{min}},R_{\yy{max}}]$ is \yy{discrete}, and $\overline{P}$ is continuous. We can have the following proposition:
\begin{theorem}\label{Theorem:LMP}
For a certain power strategy $P \in \overline{P}$, the potential function $F(A)$ , where $A =(P,R)$ satisfies LMP for $R \in \overline{R}$.
\begin{proof}
See in Appendix \ref{app:2}.
\end{proof}
\end{theorem}

This leads to the following proposition:

\begin{proposition}
Suppose that $A =(P,R)$ satisfies LMP for $R \in \overline{R}$. Then, only if   $ F(A=(P,R_x),A_{-i}) \geq F(A=(P,R_y),A_{-i})$ for all $y$, $F(A=(P,R_x),A_{-i}) \geq F(A=(P,R_y),A_{-i})$ for all $|x-y|\leq1$
\label{Propo:6}
\end{proposition}

This means that if LMP is satisfied for a discrete potential function, then the local optimality in the potential function implies global optimality. Thus, when at \ds{a} certain transmit power level, only one maximizer exists over the discrete set of data rate. The proof of the Proposition \ref{Propo:6} is shown in Appendix \ref{app:3}.

As \textbf{Theorem\ref{theorem2}} shows that when $R$ is fixed, the utility function admits one optimizer (maximizer). Hence, for exact potential game, the potential function also admits unique maximizer. Thus we have
\begin{theorem}
The maximizer in the action space $A$, namely $A^o = \argmax F(A_i,A_{-i})$ is unique and:
\begin{equation}
F(A^o,A_{-i}) = \max\{F((P^{\ast_{min}},R_{\yy{min}}),A_{-i}), ... F((P^{\ast_{\yy{max}}},R_{\yy{max}}),A_{-i}) \},
\end{equation}
where, \yy{$P^{\ast_x}$ is the maximizer when $R = R_x$}.
\end{theorem}

 To show that the equilibrium of the \ds{propose potential} game is unique, it is sufficient to prove that the set of maximizers of the potential function is singleton.  Therefore, the best-response for the \ds{potential} game $G$, $A^\ast = (P^\ast,R^\ast)$ is equals to $ A^o $, which is also unique.

\subsection{Game Efficiency}

The Nash Equilibrium solution of each individual BAN in the game $G$ is the maximization of its own utility. This leads to the problem of efficiency of the network. More specifically, for a network without \ds{a central} coordinator, the fairness of the system may degrade due to selfish actions of the players. Thus, it is important to investigate the equilibrium efficiency among the co-existing WBANs. The social welfare reflects the fairness and efficiency of the system's best response, considering all individuals utility combined.

\begin{definition}
The social welfare is defined by the aggregation sum of each WBAN's utility function as:

\begin{equation}
\Omega(\mathbf{A}) = \sum_{i=0}^n U_i(\mathbf{A})
\end{equation}

\end{definition}

The maximization of the social welfare is \ds{the} social optimum, which represents the social fairness among the system. The price of anarchy ($\mathrm{PoA}$) is used to measure the inefficiency of equilibriums among selfish players. With finite number of players in game $G$, the $\mathrm{PoA}$  \footnote{Since the Nash Equilibrium is unique, the the price of anarchy equals to the price of stability(PoS)} is defined as the ration of the highest value of social welfare (social optimum) to the NE (as the Nash Equilibrium in $G$ is unique) of the game:

\begin{equation}
\mathrm{PoA} = \frac{ \Omega(\bf{A} ^{opt})}{ \Omega(\bf{A^{\ast}})} \geq 1,
\end{equation}

where $\bf{A} ^{opt} = \argmax \Omega(\bf{A})$ is the global optimum solution.


\section{The Adaptive Backoff Game}\label{Sec:5}
By implementing the proposed back-off algorithms, the collision during the frame transmission can possibly be avoided. Because, before each transmission, each WBAN waits for a random time, based on Contention Window size. This mechanism space out repeated retransmissions of the data in each WBAN. Generally, each WBAN is able to tune their transmission probability by modifying the back-off control parameters, such as $CW_{min}$ value and maximum back-off stages ($m$ value). Therefore, \ds{each} WBAN can dynamically choose a suitable contention window size according to the contention level of current network in order to \yy{effectively} improve system performance.


However, due to the non-cooperative nature of the system, each selfish player attempts to increase its utility by increasing its transmission probability or equivalent by decreasing its contention window size. Increasing the transmission probability by one player \ds{encourages} other players to shorten their contention window sizes, which \ds{increases} collision\ds{s}, thus the delay and packets drop ratios are \yy{also increased}. Here, we proposed the Contention Game $G_{\textrm{CW}}$ based on the aforementioned MAC layer scheduling, which aim\ds{s} to balance the trade-off between packet delay and system throughput.

In the Contention Game, \ds{the} action selected by any player is their minimum contention windows size \yy{$CW_{min}$, where $CW_{min}$ is} the  action space. As \ds{described by} the Markov Model, by changing the contention window size, \ds{players} transmission probability can be adjust\ds{ed} accordingly. As, \ds{in} \ds{a} high $\mathrm{PDR}$ regime, we have the following approximation:
\begin{equation}
\tau \approx \frac{1}{p_f \cdot CW_{min} + 1}\ds{,}
\end{equation}

\yy{where $p_f$ is transmission failure probability.}



\yy{Emp\ds{i}rically}, in order to get the value of $n$, each node can measure \yy{$p_f$} and $\tau$ through several counters independently. The number of coexisting WBANs can be estimated from the following equations\cite{ghazvini2013game}:
\begin{equation}\label{eqn:number est}
\begin{aligned}
\tau_{est}=\frac{Transmitted Fragment Count}{Slot Count}\\
p_{est}=\frac{Ack Failure Count}{Transmitted Frament Count},\\
\end{aligned}
\end{equation}


where $p_{est}$ and $\tau_{est}$ denote the estimated $\tau$, failure probability respectively.

$\mathrm{Transmitted Fragment Counter}$ that counts the total number of successfully transmitted data frames,
$\mathrm{ACK Failure Counter}$ that counts the total number of unsuccessfully transmitted data frames and the $\mathrm{Slot Counter}$ that counts the total number of experienced time slots. (The historical data can be used to estimate the current parameters, and the length of how far we should trace back can be adjusted accordingly).

\begin{algorithm}
\caption{Main steps for Adaptive Backoff Game}\label{Algorithm:3}
\begin{algorithmic}[1]
\State After execute line \ref{opretion:game Start} in Algorithm \ref{Algorithm:2}
\State The Hub estimated the number of coexisting WBANs $n_{est}$ by using equation \eqref{eqn:number est}
\State After obtain $n_{est}$. The Utility Function $V(CW_{min})$ can be constructed as a function of $w$, where the $\mathrm{PDR}$ is obtained as $\mathrm{PDR}(\bf{P_{-i}^{\ast},R_{-i}^{\ast}})$ from the Nash Equilibrium in the Link Adaptation Game.
\State Determine the minimum contention window size $CW_{min}$ for that given sensor, which gives the Nash Equilibrium value of the Utility Function $V(CW_{min})$
\State \Goto{opretion:game end} in Algorithm \ref{Algorithm:2}
\end{algorithmic}
\end{algorithm}

The objective of the game is to reach a trade-off in maximizing the throughput, and \ds{minimize} delay. \ds{Following from the} analytical model, the throughput of each WBANs is positively correlated with the Goodput(S) in equation \eqref{equation:Goodput}. The average delay for a packet to be transmitted successfully is estimated as:

\begin{equation}
\begin{aligned}
D=& \sum_{i=1}^m P_{Bi} E[D_i] \\
 =& \sum_{i=1}^{m-1} [p_f^i(1-p_f)\sum_{j=0}^i (\frac{W_j+1}{2}\yy{T_{\text{slot}}} + T_{s}) ] \\
  & + p_f^m \sum_{j=0}^m  (\frac{W_j+1}{2}T_{\text{slot}}  + T_{s}),
\end{aligned}
\end{equation}

where $ E[D_i]$ is the average delay in state $i$. It is obvious that throughput, delay may have different units in different ranges, and they have to be normalized. Therefore, the utility function is defined as the following:

\begin{equation} \label{eq:CW_Ult}
V_i (CW_{min}) = d \cdot S - l \cdot D - P_{\textrm{Drop}},
\end{equation}

where $P_{\textrm{Drop}} = \tau P^{m+1}$ is the probability that \yy{a packet drops} due to exceed maximum retry limits. The weights $d,l$ can be adjusted based on different scenarios. The obtained results have shown that in game $G_{CW}$, each user improves its \ds{chance of} successful transmission by increasing transmission probability, whil\ds{st} this increase \yy{of transmission probability} causes an \ds{increase} in collision probability, as well. Such collisions will \yy{cause large delay in packets transmission} and \ds{energy} wastage led by $\mathrm{PDR}$ reduction. Thus, in case of less contending nodes(or high SINR regime), the nodes should select a smaller $CW_{min}$ as the best strategy. In the case of more contending nodes(or low SINR regime), greater $CW_{min}$ is more appropriate in order to reduce the collision probability. The game \ds{is} implemented in a similar distributed manner \ds{to} the Link Adaptation Game.

\subsection{Game \yy{P}roperty}

\yy{The existence and uniqueness of the Nash equilibrium point for the Adaptive Backoff game is guaranteed. The proof is given as follows.}

\subsubsection{Existence and uniqueness of the Nash equilibrium}

 \begin{theorem}\label{Theorem:Ult_contention}
 (Existence and Uniqueness): In each stage of the game $G_{CW}$ exists a unique Nash Equilibrium.
 \end{theorem}

\begin{proof}
Similar with the Link Adaptation Game $G$, the utility function in $G_{CW}$ is differentiable and strictly concave over the convex \yy{set} of the minimum contention window size $CW_{min}$. Therefore, according to \cite{han_game_2012}, the game $G_{CW}$ admits a unique Nash Equilibrium. The details can be found in Appendix \ref{app:1}.
\end{proof}

\subsubsection{Efficiency}
Similarly with the Link Adaptation Game the social welfare of the Adaptive Backoff Game $G_{CW}$ is defined as :
\begin{equation}
\Omega_{CW} = \sum_{i=0}^{n} V_i
\end{equation}

With finite number of players in game $G_{CW}$ the $\mathrm{PoA}$ is defined as the ration of the highest value of social welfare (global optimization) to the NE (as the Nash Equilibrium in $G_{CW}$ is unique) of the game:

\begin{equation}
\mathrm{PoA_{CW}} = \frac{ \Omega_{CW}(\bf{B} ^{opt})}{ \Omega_{CW}(\bf{B^{\ast}})} \geq 1,
\end{equation}

where $\bf{B} ^{opt} = \argmax \Omega_{CW}(\bf{B})$ is the global optimum solution.

\section{SIMULATION ANALYSIS}\label{Sec:6}

This section discusses the simulation results of our proposed MAC layer games in contrast with conventional schemes as well as game-theory based methods in the literatures. To evaluate and validate the performance of the proposed game, we compare throughput, energy efficiency and delay with $\mathrm{B}^2$RIS\cite{grassi2012b2irs} and Adaptive CSMA/CA \cite{xia2013service} in respect of varying numbers of consisting WBANs in our simulations. The value of MAC layer Parameters are listed as below, which is mainly based on the IEEE802.15.6 standard (described in section \ref{802.15.6}). In addition, a non-linear power estimation \cite{davenport2009medwin} is also made to measure the actual circuit \ds{energy} consumptions, \ds{to provide a more realistic} evaluation of the system.

\begin{table}[!htb]

\centering
\caption{MAC Parameter }
\begin{tabular}{|c|c|}
\hline
Parameters & Value \\ \hline
   Superframe Length & 80 $ms$ \\ \hline
Allocated Time Slot Length  & 0.312 $ms$   \\ \hline
 Minimum Data Rate &   \SI{25.6}{\kilo\bit\per\second}  \\ \hline
 Maximum Data Rate &  1.28 Mbits $\text{s}^{-1}$  \\ \hline
Payload & 175 bytes   \\ \hline
 $N_{\textrm{MAC Hdr}}$ & 24 bits    \\ \hline
  $N_{\textrm{MAC Ftr}}$ & 24 bits    \\ \hline
    $N_{\textrm{Beacon}}$ & 20 bytes    \\ \hline
        $N_{\textrm{ACK}}$ & 10 bytes    \\ \hline
     $CW_{\textrm{min}}$ & 43.75 ms    \\ \hline
      MaxBackoffLimits (m) & 4    \\ \hline
\end{tabular}
\label{tab:BPSKTable}
\end{table}

\subsection{Scenario 1}
In this scenario, a realistic measurement \cite{WBAN2017Measure} using small body-mounted "channel sounder" radios that operated at 2.36 Ghz is ad\ds{o}pted in the simulation. The measurement set contains both inter-WBAN and intra-WBAN channels of the co-exi\ds{s}ting WBANs, which are measured on human wearers in many different environments, involving subjects doing a mix of distinctive everyday activities. It should also be noted that the measurements are re-sampled by the parameters above, thus the channel attenuations remain constant in each superframe \yy{\cite{DBLP:journals/corr/abs-1305-6992}}.

\begin{figure}
\centering
\includegraphics[width=10cm,height=6cm]{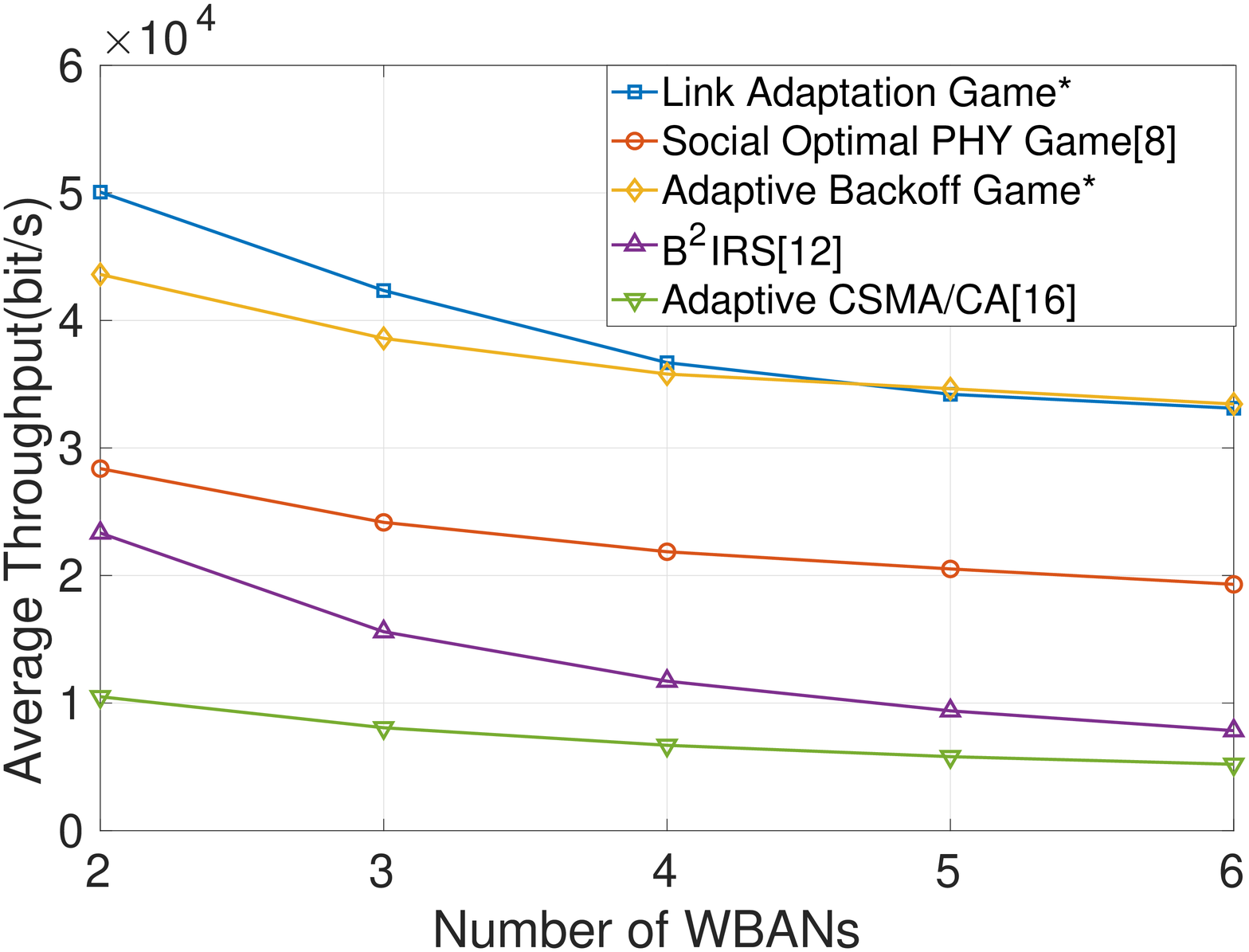}
\caption{Throughput Performance \yy{of the proposed games compared to other methods} under \yy{Realistic} Measurement Channel Sets}
 \label{fig:TP_M6}
\end{figure}

\begin{figure}
\centering
\includegraphics[width=10cm,height=6cm]{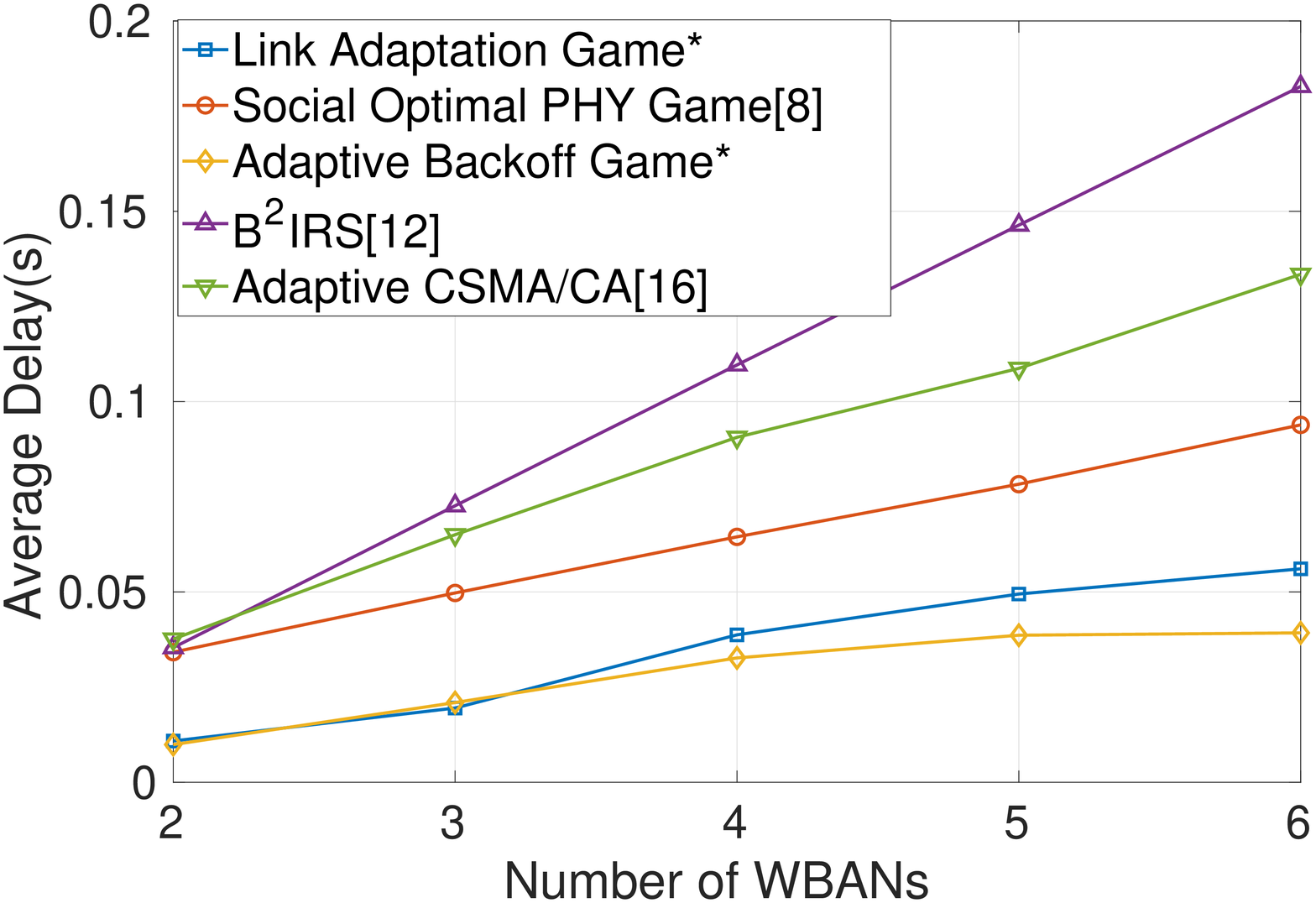}
\caption{Delay Performance \yy{of the proposed games compared to other methods} under \yy{Realistic} Measurement Channel Sets}
 \label{fig:Delay_M6_new}
\end{figure}

\begin{figure}
\centering
\includegraphics[width=10cm,height=6cm]{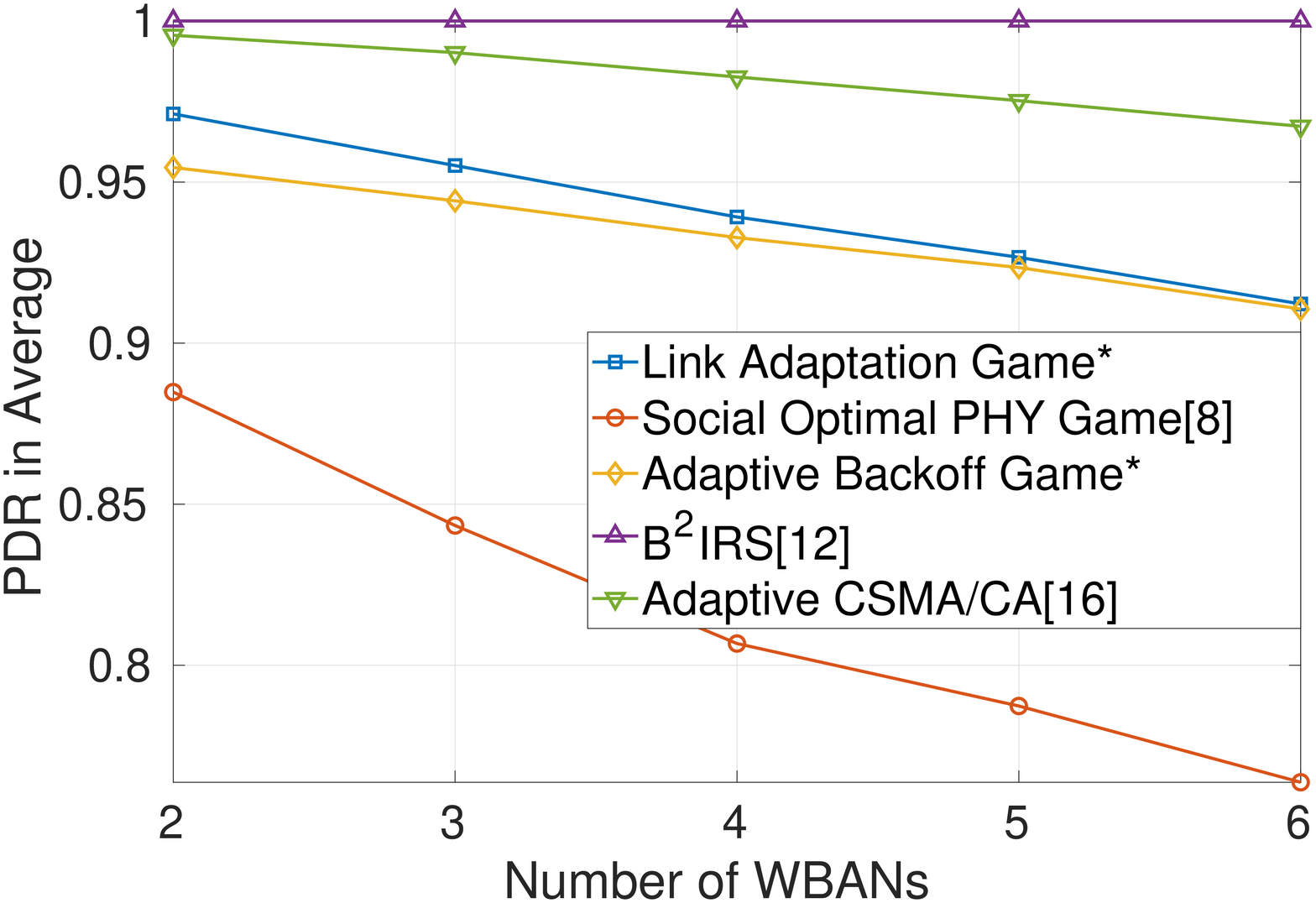}
\caption{PDR performance \yy{of the proposed games compared to other methods} under \yy{Realistic} Measurement Channel Sets}
 \label{fig:PDR_M6}
\end{figure}

\begin{figure}
\centering
\includegraphics[width=10cm,height=6cm]{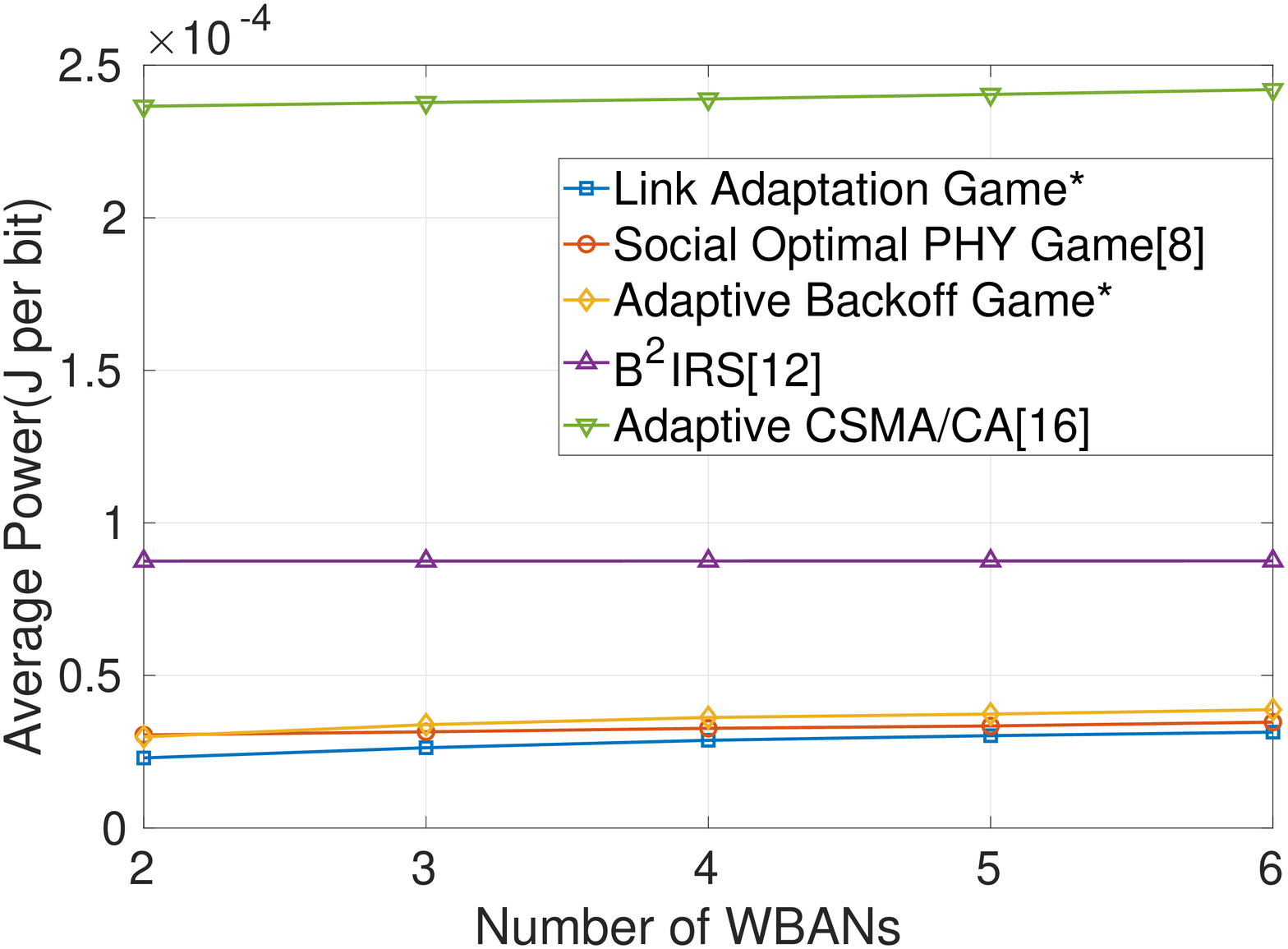}
\caption{\yy{Circuit} power consumption \yy{of the proposed games compared to other methods} under \yy{Realistic} Measurement Channel Sets}
 \label{fig:Power_M6}
\end{figure}

Figures \ref{fig:TP_M6} and \ref{fig:Delay_M6_new} show the \ds{variation in t}hroughput and delay\ds{, respectively,} with different numbers of WBAN co-existing. Fig. \ref{fig:TP_M6} shows that the throughput is much reduced \ds{with} more WBANs in the system as the collision probability is increased. In $\mathrm{B}^2$IRS, the  packets are rescheduled  in a collision-free manner, hence no interference occurs.
The proposed games provide lower delay and higher throughput compare\ds{d} with other methods. The Adaptive Backoff Game outperforms the Link Adaptation Game in terms of delay and average throughput \ds{although the Link Adaptation Game is superior in terms of power consumption and $\mathrm{PDR}$}. \ds{H}owever, when larger number of WBANs are co-located the advantage of the contention window length game will be more obvious, which is shown in the next section.

\subsection{Scenario 2}

\ds{T}he measurement set \ds{used} can only provide channel gain of up to 6 WBANs coexisting, \yy{computer-simulated} channels are also needed to obtain performance analysis under crowded environment where \ds{many} more WBANs are co-located.

\begin{figure}
\centering
\includegraphics[width=10cm,height=6cm]{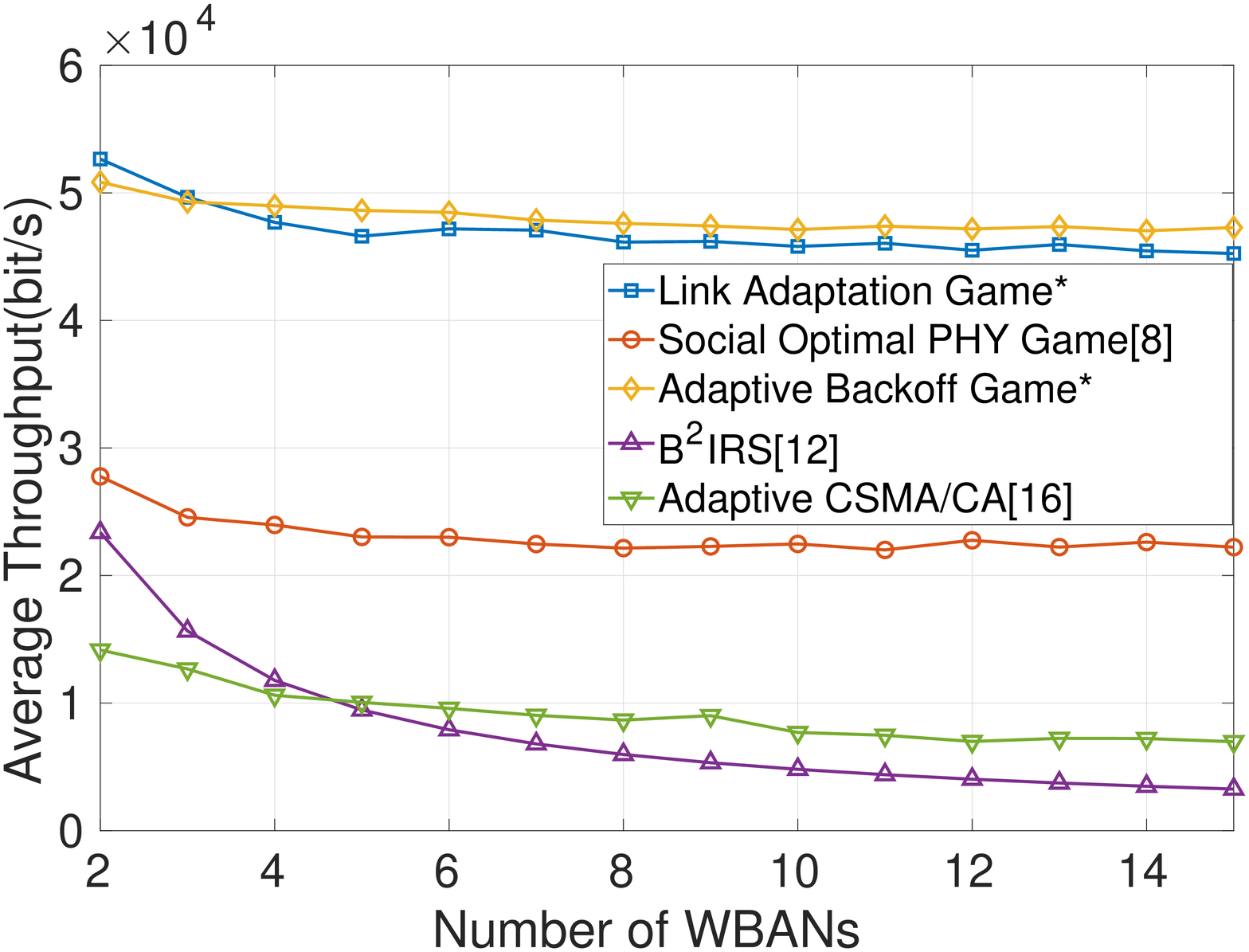}
\caption{Throughput performance \yy{of the proposed games compared to other methods} under \yy{Simulated} Channel Sets}
 \label{fig:TP_M15}
\end{figure}

\begin{figure}
\centering
\includegraphics[width=10cm,height=6cm]{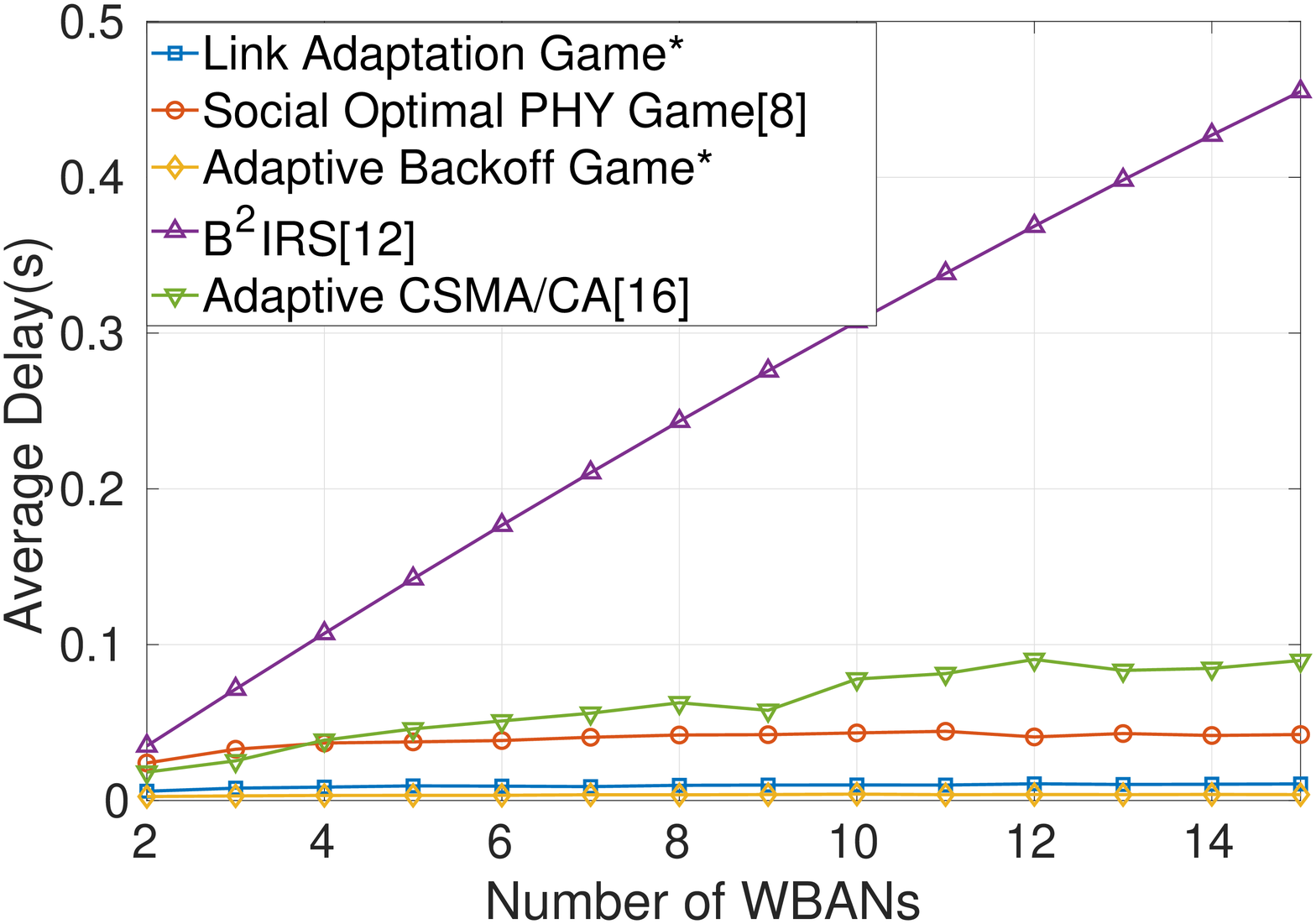}
\caption{Delay performance \yy{of the proposed games compared to other methods} under \yy{Simulated} Channel Sets}
 \label{fig:DL_M15}
\end{figure}
\begin{figure}
\centering
\includegraphics[width=10cm,height=6cm]{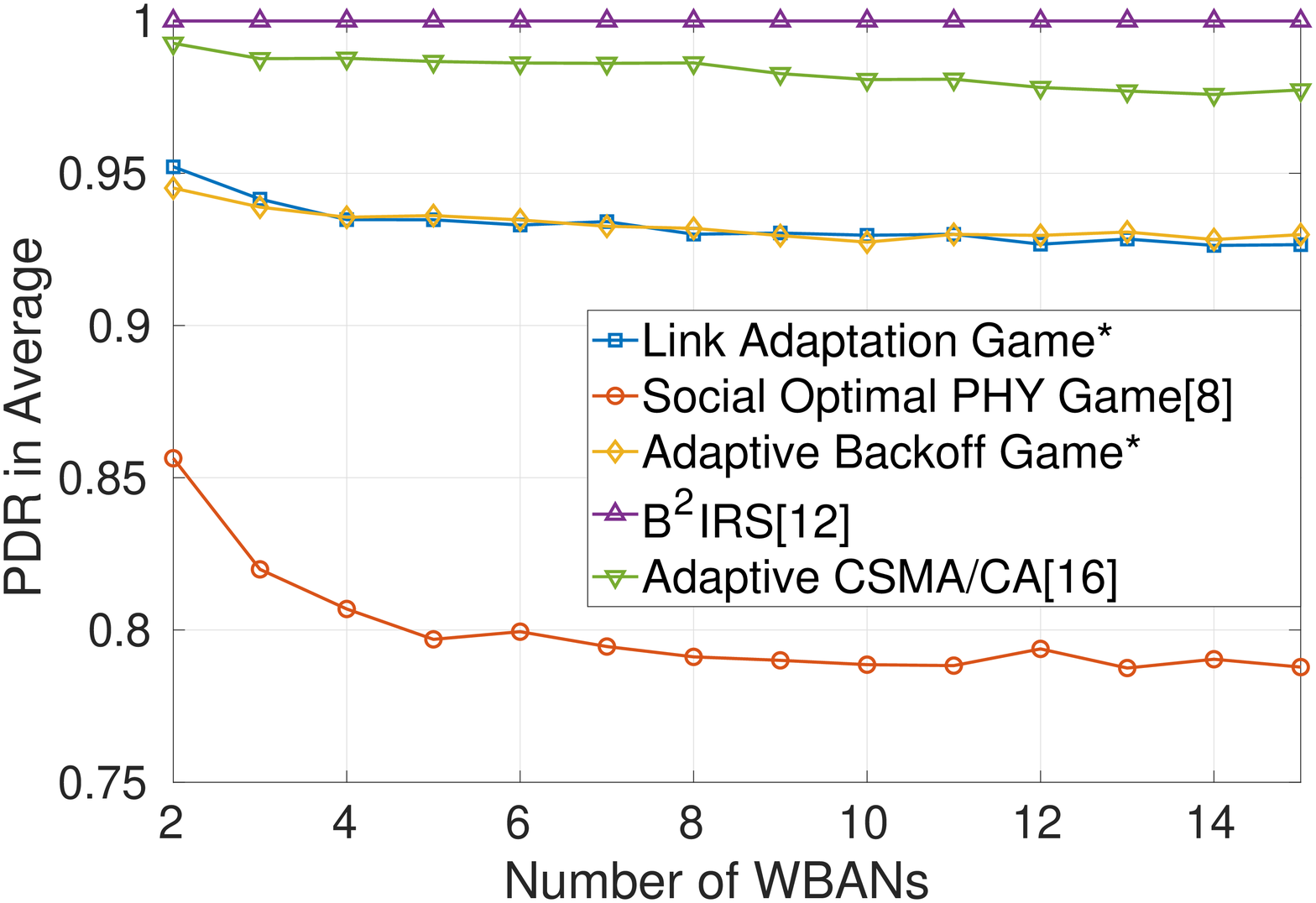}
\centering
\caption{PDR performance \yy{of the proposed games compared to other methods} under \yy{Simulated} Channel Sets}
 \label{fig:TP_M15}
\end{figure}

In the simulation, both intra-BAN and inter-BAN channels are modeled. It is assumed that up to 15 WBANs with the same topology are coexisting and moving randomly within a $6 \times 6 \ m^2$ square area. The walking speed of the WBAN wearer is modeled as  $ 0.5 \pm 0.1 $ ~m/s, which is updated every 1 ms. The channel attenuation is modeled as $h_i^j=A_t(d_o/d_i^j)^{(2.5/2)}A_{SE}A_{SC}$,  where the path loss exponent is 2.5. $d_i^j$ represents the distance between BAN $i$ and $j$, and the reference distance $d_o = 5m$ corresponds to a channel attenuation of 50~dB. The shadowing effect $A_{SE}$ is assumed to be 42dB, and a Jakes model with Doppler spread of 1.1 Hz as the $\mathcal{CN}(0,1)$ Rayleigh distrib\ds{ed} small scale fading $A_{SC}$ between WBANs. Gamma fading with a mean 65 dB attenuation, a shape parameter of 1.31, and a scale parameter of 0.562 is employed for the on-body channels.

Basically, for the proposed methods, less than $10\%$ of the packets are blocked. Again, $\mathrm{PDR}$ performances of $B^2$IRS  and adaptive CSMA/CA are better than others, because of the low collision probabilities in these two methods.

In Figure \ref{fig:DL_M15}, th\ds{ere is increasing delay} with increase in the  number of co-located WBANs. Amongst three game-theory-based methods, Social Optimal PHY Game has the highest delay due to large re-transmissions. The Adaptive Backoff Game provides smallest packet delay, and the delay time is increased a small amount at higher interference regime. At the same time, \yy{$\mathrm{B}^2$}IRS has the largest delay due to complexity of beacon re-scheduling when greater number of WBANs are co-existing.

As described in Figure \ref{fig:TP_M15}, when more than 4 WBANs co-exiting, the Adaptive Backoff Game can provides higher throughput. Both two proposed methods have significantly larger throughput respect to other methods, as higher date rate are more preferable in the game at relative better channel conditions.

\begin{figure}
\centering
\includegraphics[width=10cm,height=6cm]{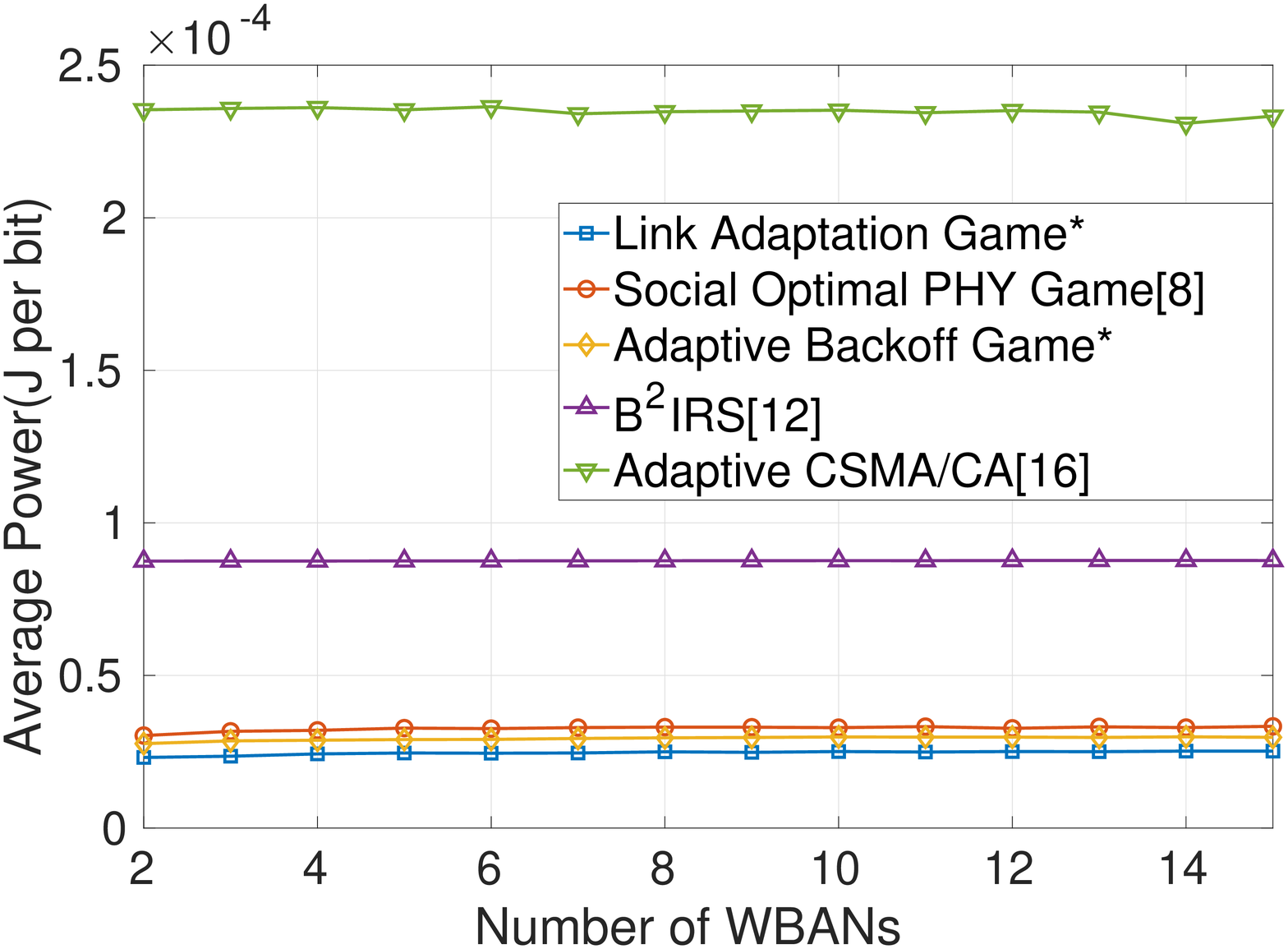}
\caption{Circuit power consumption \yy{of the proposed games compared to other methods} under \yy{Simulated} Channel Sets}
 \label{fig:P_Con_M15}
\end{figure}
By applying the non-linear circuit power mapping \cite{davenport2009medwin}, the circuit power consumptions can be estimated. The proposed Link Adaptation Game method provides the lowest power consumption in terms of Joules per bit (J/bit). The $\mathrm{B^2}$IRS consumes 10 times more power per bit transmitted. However, the Adaptive Backoff Game uses slightly larger power than the Link Adaptation Game. This is because, in the Adaptive Backoff Game, when the contention windows size is small, more packet\ds{s} are transmitt\ds{ed} concurrently, hence the transmitter uses larger power to achieve reasonable $\mathrm{PDR}$.

\subsection{Efficiency of the Game}

\yy{Social efficiency is \ds{a} key measurement for a reliable and efficient game design.} We evaluate the Price of Anarchy \yy{($\mathrm{PoA}$)} of the two prosed games by implementing a Monte Carlo simulation on time varying channels, where \ds{an} interior point \yy{approach} is applied to find the centralized (global) optimum of the social welfare. Although, different WBANs decide their actions in different time slots, we assume the global maximization point is solved instantaneously \ds{despite} MAC layer scheduling.
  \begin{figure}
\centering
\includegraphics[trim=0cm 0cm 0cm 0cm,clip=true,width=10cm,height=6cm]{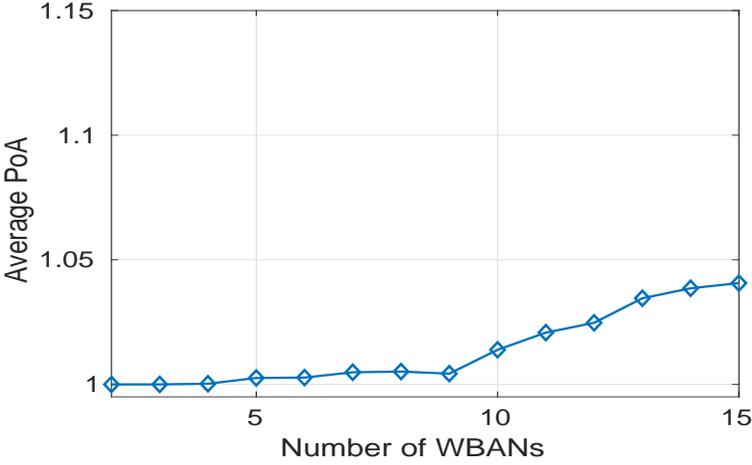}
\caption{$\mathrm{PoA}$ in Link Adaptation Game}
 \label{fig:PoA_R}
\end{figure}

\begin{figure}
\centering
\includegraphics[trim=0cm 0cm 0cm 0cm,clip=true,width=10cm,height=6cm]{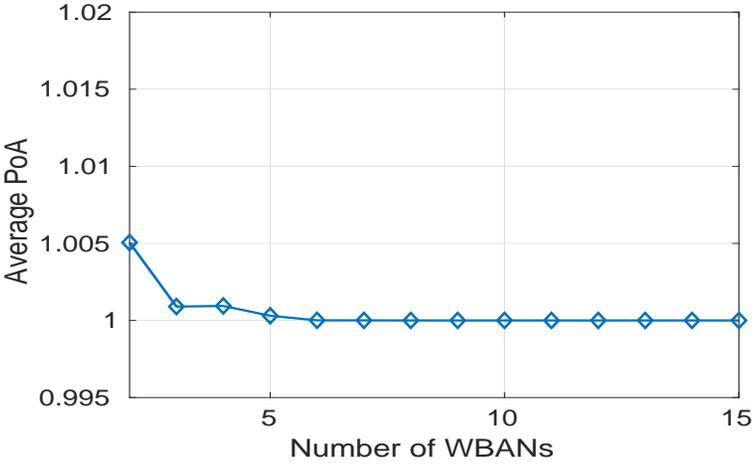}
\caption{$\mathrm{PoA}$ in Adaptive Backoff Game}
 \label{fig:PoA}
\end{figure}

\ds{Figs. \ref{fig:PoA_R} and \ref{fig:PoA}} above illustrate the $\mathrm{PoA}$ for different number\ds{s} of co-existing WBANs. It can be seen that the loss due to decentralization is relatively small as $\mathrm{PoA} \to 1$. Also, in general , the system waste less than $10\%$ of their welfare in terms of utility for not being coordinated from both Fig. \ref{fig:PoA_R} and Fig. \ref{fig:PoA}. Meanwhile, we introduce a new metrics $L$,$L_{CW}$:
\begin{equation}
L = \exp(\frac{\Omega(\bf{A^{\kappa}})}{\Omega(\bf{A^{\ast}})}),
\end{equation}
and,
\begin{equation}
L_{CW} = \exp(\frac{\Omega_{CW}(\bf{B^{\kappa}})}{\Omega_{CW}(\bf{B^{\ast}})}),
\end{equation}
where $\bf{A^{\kappa}} \in \overline{A}, \bf{A^{\kappa}}= \argmin \Omega(\bf{A})$, and $ \bf{B^{\kappa}} \in \overline{B}$, $\bf{B^{\kappa}}= \argmin \Omega_{CW}(\bf{B})$. $L$, $L_{CW}$ correlated with the ratio of the worst value of the social welfare and the maximum value of the social welfare in game $G$ and $G_{CW}$ respectively (we take the exponential as sometime the value of utility function can be negative). These two metrics represent the gap between the system's best possible performance and the worst case scenario. The comparison between $L$,$L_{CW}$ and $\exp(\frac{1}{\mathrm{PoA}})$ and $\exp(\frac{1}{\mathrm{PoA}_{CW}})$ provide\ds{s} some insight of how stable the Nash equilibrium is \ds{across iterations of the game}. Hence, in Table \ref{tab:L}, we illustrate the comparison between the mean value of $L$ and $\exp(\frac{1}{\mathrm{PoA}})$ when different number of WBANs are co-located. Meanwhile, the comparison between $L_{CW}$ and $\exp(\frac{1}{\mathrm{PoA}_{CW}})$ is depicted in in Table \ref{tab:L_CW}.

\begin{table}
\centering
\caption{Comparison of $\exp(\frac{1}{\mathrm{PoA}})$ and the mean value of $L$  }
\begin{tabular}{|c|c|c|}
\hline
   No. of WBANs & $L$ & $\exp(\frac{1}{\mathrm{PoA}})$\\ \hline
 $2$  & 6.832e-32  &2.718		  \\ \hline
 $5$  & 4.252e-32  & 2.704 \\ \hline
 $8$  & 2.291e-32 & 2.562 \\ \hline
 $11$ & 2.728e-32 & 2.549 \\ \hline
 $14$ & 1.141e-32 & 2.542 \\ \hline
\end{tabular}
\label{tab:L}
\end{table}

\begin{table}
\centering
\caption{Comparison of $\exp(\frac{1}{\mathrm{PoA}_{CW}})$ and the mean value of $L_{CW}$ }
\begin{tabular}{|c|c|c|}
\hline
   No. of WBANs & $L_{CW}$ & $\exp(\frac{1}{\mathrm{PoA}_{CW}})$\\ \hline
 $2$  & 2.2445 & 2.7056		  \\ \hline
 $5$  & 2.0317  & 2.7175 \\ \hline
 $8$  & 1.5793 & 2.7183 \\ \hline
 $11$ & 1.5746 & 2.7183 \\ \hline
 $14$ & 1.5864 & 2.7183 \\ \hline
\end{tabular}
\label{tab:L_CW}
\end{table}

Comparing with $\yy{\exp}
(\frac{1}{\mathrm{PoA}})$, $L$ is much more smaller. It is because the social welfare varies a significant amount over the action space. Thus, in Link Adaptation Game, the system is socially stable as the deviation from the social optimum solution is small. On the other hand, in contention game, the values of $L_{CW}$ are relatively large. However, it can be seen that, the $\mathrm{PoA}_{CW}$ is close to $1$, which represent\ds{s complete} stability of \ds{the} game.

\section{Conclusion}
\ds{An} insightful game theory model \ds{has been proposed} to adaptively adjust transmit power and data rate to mitigate inter-WBAN interference level \ds{while reducing} overall \ds{energy} consumption. The model is based on a novel contention-based MAC layer \ds{protocol} with special back-off mechanism, which \ds{reduces} packet collision probability. Besides, another game that can optimize the length of back-off is proposed in order to reduce average delay and increase system throughput. To compare  with some alternative state-of-art approaches, we conducted several simulations \ds{for both empirical and simulated channels}. The simulation results reveal that the proposed methods outperform \ds{state-of-art}, in terms of energy consumption, and overall \ds{quality-of-service (QoS).}

Although, both of the proposed methods have $\mathrm{PoA} \rightarrow 1$ \ds{and are very close to the social optimum, such} optimum is not guaranteed \ds{at the game's} Nash Equilibrium. \ds{If such a guarantee could be provided that would be a further significant advance}. Furthermore, as described above, in the \ds{MAC} layer protocol, the length of back-off period is randomly \ds{chosen}. \ds{However, even} without prediction of the global channel state (channel gains), such randomness can provide efficient collision avoidance. \ds{But more precise allocation} can be made if channel \ds{state} information can be accurately predicted. Thus, \ds{current further work includes} channel prediction, \ds{potentially} improv\ds{ing} overall throughput by offering a \ds{reduction in} packet collision rate. Meanwhile, machine learning, \ds{such as} reinforcement learning\ds{,} is also regarded as \ds{a further potential improvement}.

\begin{appendices}
\section{Proof of Theorem \ref{Theorem:Ult_contention}}\label{app:1}

Firstly, the utility function \eqref{eq:CW_Ult} denoted as:
\begin{equation}
V=d \cdot S-c \cdot D - P_{\text{\yy{Drop}}}
\end{equation}
It can be seen that the first derivative of the utility function is continuous:
\begin{equation}
\begin{aligned}
\yy{d \cdot} \frac{\partial S}{\partial CW} &= \frac{\partial S}{\partial \tau}   \frac{\partial \tau}{\partial CW}\\[10pt]
&= \yy{\frac{\partial \tau}{\partial CW}} \cdot \frac{(1-P)\yy{T_{\text{payload}}}[\yy{T_{\text{slot}}}+\tau(\yy{T_{\text{s}}}-\yy{T_{\text{slot}}})]  - (\yy{T_{\text{s}}}-\yy{T_{\text{slot}})}\tau(1-P)\yy{T_{\text{payload}}}) } {[\yy{T_{\text{slot}}}+\tau(\yy{T_{\text{s}}}-\yy{T_{\text{slot}}})]^2}\\
&= -(aCW +1)^{-2} \frac{\yy{d}(1-P)\yy{T_{\text{payload}}}\yy{T_{\text{slot}}}}{[\yy{T_{\text{slot}}}+\tau(\yy{T_{\text{s}}}-\yy{T_{\text{slot}}})]^2}
\end{aligned}
\end{equation}
Hence, according to \cite{10.2307/1911749}, in \ds{this} contention game, \ds{the} Nash equilibrium exists.
\begin{equation}
\begin{aligned}
\yy{d \cdot} \frac{\partial^2 S}{\partial CW^2} = &
 \frac{2\yy{d}(1-P)\yy{T_{\text{payload}}}\yy{T_{\text{slot}}}([\yy{T_{\text{slot}}}+\tau(\yy{T_{\text{s}}}-\yy{T_{\text{slot}}})]^2)\tau^2(\yy{T_{\text{s}}}-\yy{T_{\text{slot}}})}{[\yy{T_{\text{slot}}}+\tau(\yy{T_{\text{s}}}-\yy{T_{\text{slot}}})]^4}
\\
& + 2\tau^3  \frac{\yy{d}(1-P)\yy{T_{\text{payload}}}\yy{T_{\text{slot}}}}{[\yy{T_{\text{slot}}}+\tau(\yy{T_{\text{s}}}-\yy{T_{\text{slot}}})]^2}
\end{aligned}
\end{equation}
Meanwhile, delay is a linear function of the contention window size, thus $\frac{\partial^2 D}{\partial CW^2} =0$.

\begin{equation}
\frac{\partial^2 P_{\text{Drop}}}{\partial CW^2} = 2\tau^3P^{m+1}\ds{,}
\end{equation}
which should always be positive. Therefore, the second order derivative of the utility function can be obtained as follows:

\begin{equation}
\begin{aligned}
\frac{\partial^2 v}{\partial CW^2} = &-2\tau^3P^{m+1} + 2\tau^3 \frac{\yy{d}(1-P)\yy{T_{\text{payload}}}\yy{T_{\text{slot}}}}{[\yy{T_{\text{slot}}}+\tau(\yy{T_{\text{s}}}-\yy{T_{\text{slot}}})]^2}  \\
& + \frac{2\yy{d}(1-P)\yy{T_{\text{payload}}}\yy{T_{\text{slot}}}([\yy{T_{\text{slot}}}+\tau(\yy{T_{\text{s}}}-\yy{T_{\text{slot}}})]^2)\tau^2(\yy{T_{\text{s}}}-\yy{T_{\text{slot}}})}{[\yy{T_{\text{slot}}}+\tau(\yy{T_{\text{s}}}-\yy{T_{\text{slot}}})]^4}
\end{aligned}
\end{equation}

When $d$ \yy{is set} in a reasonable range, $\frac{\partial^2 v}{\partial CW^2}$ \yy{will always be} less than zero. Hence, the utility function is concave.

\section{Proof of Theorem \ref{Theorem:LMP}}\label{app:2}
In WBAN $i$, for any $x,y,z \in X$, $\lVert x-y \rVert = 2, \lVert x-z \rVert=\lVert z-y \rVert = 1 $, we have
\begin{equation}
\begin{aligned}
F(A_i=(P,R_x),A_{-i}) &= -cP_x^g - q\frac{1}{R_x}
+\frac{\alpha(R_x) \cdot \gamma  ^{\beta(R_x)} }{2}
+  \sum_{j \neq i} -cP_j^w
- q\frac{1}{R_j} +
\frac{\alpha(R_j) \cdot \gamma ^{\beta(R_j)} }{2}
\\
&= C(P_x)+
G(R_x)+
\yy{H(A_i)}
+ Q(A_{-i}),
\end{aligned}
\end{equation}
where $H(A_i) = \yy{H((P,R_x))}  = \frac{\alpha(R_x) \cdot \gamma  ^{\beta(R_x)} }{2} $, $Q(A_{-i})= \sum_{j \neq i} -cP_j^g
- q\frac{1}{R_j} +
\frac{\alpha(R_j) \cdot \gamma ^{\beta(R_j)} }{2}$. Meanwhile, $F(A_z,A_{-i})$ and $F(A_y,A_{-i})$ are expressed in a similar way.

Obviously, $G(R)$ is strictly concave, and increase with $R$. Hence, $2G(R_z)>G(R_x)+G(R_y)$. Also, as shown in table \ref{tab:BPSKTable}, $0>2\alpha(R_z) > \alpha(R_x)+\alpha(R_y)$. When at high $\mathrm{PDR}$ regime, where $\gamma > 0 dB$, as $\beta(R_x)$ increases with $R$ and $\gamma ^ \beta(R_x)>0$ is convex:
\begin{equation}\label{eq:concave_pf_1}
\begin{aligned}
(\gamma ^\beta(R_y) - \gamma ^\beta(R_z))- (\gamma ^\beta(R_z) - \gamma ^\beta(R_x)) > 0\\
\alpha(R_x) \cdot \{(\gamma ^\beta(R_y) - \gamma ^\beta(R_z))- (\gamma ^\beta(R_z) - \gamma ^\beta(R_x))\} < 0\\
\end{aligned}
\end{equation}
Also,
\begin{equation}\label{eq:concave_pf_2}
\begin{aligned}
\gamma ^b(R_X) \cdot \{(\alpha(R_x) - \alpha(R_z))- (\alpha(R_z) - \alpha(R_y))\} < 0 \\
\end{aligned}
\end{equation}
Combining equation \ref{eq:concave_pf_2} and equation \ref{eq:concave_pf_1}:
\begin{equation}
\begin{multlined}
\alpha(R_x) \cdot \{(\gamma ^\beta(R_y) - \gamma ^\beta(R_z))- (\gamma ^\beta(R_z) - \gamma ^\beta(R_x))\} \\+
\gamma ^\beta(R_X) \cdot \{(\alpha(R_x) - \alpha(R_z))- (\alpha(R_z) - \alpha(R_y))\} <0,
\end{multlined}
\end{equation}
and
\begin{equation}
\begin{multlined}
\alpha(R_x) \cdot \{(\gamma ^b(R_y) - \gamma ^b(R_z))- (\gamma ^b(R_z) - \gamma ^b(R_x))\} \\
+
\gamma ^b(R_X) \cdot \{(\alpha(R_x) - \alpha(R_z))- (\alpha(R_z) - \alpha(R_y))\}\\
-(\alpha(R_x) - \alpha(R_z))(\gamma ^b(R_y) - \gamma ^b(R_z))
- (\alpha(R_z) - \alpha(R_y))(\gamma ^b(R_z) - \gamma ^b(R_x)) <0
\end{multlined}
\end{equation}
Thus:
\begin{equation}
\begin{aligned}
\alpha(R_y)\gamma ^\beta(R_y) -2 \alpha(R_z)\gamma ^\beta(R_z) + \alpha(R_x)\gamma ^\beta(R_x) <0
\end{aligned},
\end{equation}

which means $2H(A_i=(P,R_z))\geq H(A_i=(P,R_y))+H(A_i=(P,R_x))$, such that:
\begin{equation}
2F(A_i=(P,R_z),A_{-i}) > F((P,R_x),A_{-i})+F((P,R_y),A_{-i})
\end{equation}
Therefore, the potential function $F$ satisfies Lemma\ref{lemma1}.

\section{Proof of Proposition \ref{Propo:6}}\label{app:3}

Let $x$ satisfy $F((P,R_x),A_{-i}) \geq F((P,R_y),A_{-i})$ for all $y$ with $|x-y|\leq1$. For y with $d=|x-y|\geq2$, we can make a sequence ${x}^d_{k=0}$ such that $x^0=x$ and $x^d=y$ with the following steps:
\begin{equation}
x^{k+1}\in \argmax_{|x-z|=1,|y-z|=d-k-1}  F((P,R_z),A_{-i})
\end{equation}

Suppose $|x^k-z|=|x^{k+2}-z|=1$, then we have $d-k=|x^k-y|=|x^k-z+z-y|\leq |z-y|+1$. Meanwhile, $|z-y| = |z-x^{k+2}+x^{k+2}-y| \leq |z-x^{k+2}|+|x^{k+2}-y| = d-k-1$. Therefore, $d-k-1 = |z-y|$, which gives the following equation:
\begin{equation}
F((P,R_{x^{k+1}}),A_{-i})  = \max_{|x^k-z|=1,|z-y|=d-k-1} F((P,R_{x^{z}}),A_{-i}) \geq \max_{|x^k-z|= x^{k+2}-z|=1} F((P,R_{x^{z}}),A_{-i})
\end{equation}

Since $F$ satisfies LMP as proved above, for $0 \leq k\leq d-2$:

\begin{equation}
 \max_{|x-z|=|z-y| = 1} F((P,R_{x^{k+1}}),A_{-i}) = \left\{
\begin{aligned}
&\geq F((P,R_{x^{k}}),A_{-i}) = F((P,R_{x^{k+2}}),A_{-i})\\&\qquad \qquad \mathrm{,if} \,  F((P,R_{x^k}),A_{-i})= F((P,R_{x^{k+2}}) \\
& > \min\{F((P,R_{x^k}),A_{-i}),F((P,R_{x^{k+2}}),A_{-i})\}   \mathrm{,otherwise}
\end{aligned}
\right.
\end{equation}

Since $|x^0-x^1|=1, F((P,R_{x^0}),A_{-i})\geq F((P,R_{x^1}),A_{-i})$, Also, by using the above properties, for all $k$ we have: $F((P,R_{x^k}),A_{-i}) \geq F((P,R_{x^{k+1}}),A_{-i})$. Thus, by induction we can have $F((P,R_{x}),A_{-i})=F((P,R_{x^0}),A_{-i}) \geq F((P,R_{x^1}),A_{-i}) ... F((P,R_{x^{d-1}}),A_{-i}) \geq F((P,R_{x^d}),A_{-i}) = F((P,R_{y}),A_{-i})$.
\end{appendices}

\bibliographystyle{ACM-Reference-Format}
\bibliography{ms}

\end{document}